\documentclass[onecolumn
 reprint,
 amsmath,amssymb,
 aps,
]{revtex4-2}
\usepackage{graphicx}
\usepackage{amsmath,amssymb,amsthm, physics,dsfont}
\usepackage{tikz}
\usetikzlibrary{quantikz}
\usepackage{hyperref}
\newcommand{\be}{\begin{equation}}
\newcommand{\ee}{\end{equation}}
\newcommand{\bea}{\begin{eqnarray}}
\newcommand{\eea}{\end{eqnarray}}
\newcommand{\ba}{\begin{eqnarray*}}
\newcommand{\ea}{\end{eqnarray*}}
\let\ketbra\undefined
\newcommand{\ketbra}[2]{|#1\rangle\!\langle #2 |}

\usepackage[english]{babel}
\newtheorem{theorem}{Theorem}
\newtheorem{lemma}[theorem]{Lemma}
\newtheorem{definition}[theorem]{Definition}
\newtheorem{corollary}[theorem]{Corollary}

\begin{document}

\title{Multivariable QSP and Bosonic Quantum Simulation using Iterated Quantum Signal Processing}

\author{Niladri Gomes}
\affiliation{Lawrence Berkeley National Laboratory, Berkeley, CA 94720 USA}
\author{Hokiat Lim}
\affiliation{{Department of Physics, University of Toronto, Toronto, ON, M5S 2E4, Canada}}
\author{Nathan Wiebe }
\affiliation{Department of Computer Science, University of Toronto, Toronto, ON, M5S 2E4, Canada}
\affiliation{Pacific Northwest National Laboratory, Richland, WA, 99354, USA}
\affiliation{Canadian Institute for Advanced Studies, Toronto, ON, M5G 1M1, Canada}

\begin{abstract}
    We provide in this work a form of Modular Quantum Signal Processing that we call iterated quantum signal processing.  This method recursively applies quantum signal processing to the outputs of other quantum signal processing steps, allowing polynomials to be easily achieved that would otherwise be difficult to find analytically.  We specifically show by using a squaring quantum signal processing routine, that multiplication of phase angles can be approximated and in turn that any bounded degree multi-variate polynomial function of a set of phase angles can be implemented using traditional QSP ideas.  We then discuss how these ideas can be used to construct phase functions relevant for quantum simulation such as the Coulomb potential and also discuss how to use these ideas to obviate the need for reversible arithmetic to compute square-root functions needed for simulations of bosonic Hamiltonians.
\end{abstract}

\maketitle
\section{Introduction}
Quantum signal processing (QSP) \cite{Low_2019,low2016methodology,gilyen2019quantum,motlagh2023generalized,rossi2022multivariable,Martyn_2021,alexis2024quantum} has emerged as a highly effective algorithmic technique within quantum computing.
The central idea of quantum signal processing is to provide a method that gives a polynomial transformation of a blackbox input rotation by interspersing sequences of this rotation with predefined rotations.  This idea has been applied to a number of significant applications such as optimal Hamiltonian simulation methods~\cite{low2019hamiltonian,gilyen2019quantum,Martyn_2021,berry2024doubling} and has emerged as a more qubit efficient method for linear systems algorithms than older linear combinations of unitaries based approaches~\cite{childs2012hamiltonian,childs2017quantum,gilyen2019quantum}.

The original version of quantum signal processing enables either even or odd polynomial transformations to be carried out on the matrix elements of the signal unitary \cite{low2016methodology}, but this has been generalized in a number of important ways.  The parity restrictions on the polynomial have been lifted by~\cite{motlagh2023generalized} and to the more general setting of transforming the singular values or eigenvalues of arbitrary matrices \cite{Gily_n_2019,dong2022ground,low2024quantum} and recently has been generalized to the case of operators acting on infinite dimensional vector spaces \cite{dong2022infinite} and also has been generalized to a composable form of QSP in~\cite{rossi2023modular}. 

There are several important challenges that remain in quantum signal processing.  One of the most notable issues arises from the fact that multivariate analogues of quantum signal processing are challenging to both design and/or find rotation angles for~\cite{Rossi_2022,rossi2023modular, rossi2022multivariable}. Further, even for single variable functions it can be challenging to find explicit polynomials that adequately approximate the target function while conforming to the assumptions needed for quantum signal processing.  

We take a step towards addressing both of these issues here by providing a paradigm for combining QSP sequences together to build polynomial transformations that would be difficult to achieve directly for the single qubit (${\rm SU}(2)$) Quantum Signal Processing~\cite{low2016methodology}.  We first show that we can recursively construct a form of multi-variable QSP within our framework by chaining together QSP sequences to add the rotation angles and then to apply QSP to multiply the rotation angles through building a QSP that we employ to square the sum of the rotation angles and then subtract off all terms save the product.  This gives us an elementary approach to implement any multivariate function of a set of single qubit unitaries.

The second application that we consider involves
applying our method to simulate bosonic systems. It has been shown that bosonic Hamiltonians can be simulated using qubitization and linear combination of unitaries \cite{peng2023quantum}. In spite of being sparse, simulation of bosonic Hamiltonians using these methods requires a large number of ancillary qubits and running a large amount of arithmetic (performing the square root of the matrix elements, as we will see later) which make the algorithms extremely challenging for near term quantum computers. In our work, we take advantage of the sparsity and use iterated QSP to implement the complex arithmetic using only two ancillary qubits.  In contrast, existing approaches to computing square roots on quantum computers using reversible circuitry can take thousands of qubits~\cite{haener2018quantum}.  This potentially opens up the possibility of near-term quantum simulation of bosonic systems with cutoffs in occupation greater than $1$.

The layout of the paper is as follows.  In Section~\ref{sec:background} we review the QSP formalism and discuss the logarithmic block-encoding technique that we use to construct a block encoding of the generator of a unitary.  We then discuss the application of these ideas to provide a specific form of Hierarchical QSP that we refer to as iterated QSP in Section~\ref{sec:iterQSP}.  This will be needed in our subsequent work to allow us to easily perform a polynomial transformation on the generator of the unitary in question.  Section~\ref{sec:MQSP} shows how to use the building blocks of QSP to implement a multivariate function of a set of rotations.  We also apply this methodology to provide a qubit and memory efficient method to simulate exponentials of bosonic creation/annihilation operators in Section~\ref{sec:boson} before concluding.

\section{Background} \label{sec:background}
There are multiple related quantum signal processing (QSP) conventions in the literature. For this work, we will employ the $(W_{Z}, S_{X}, \langle 0 | \cdot | 0 \rangle)$ and $(W_{X}, S_{Z}, \langle + | \cdot | + \rangle)$-QSP conventions which we now define.
\begin{definition}[QSP conventions] 
\label{def:qsp conventions}
\begin{enumerate}
Let $W_Z$ and $W_X$ be denoted as signal unitaries that are defined by
\item Under the $(W_{Z}, S_{X}, \langle 0 | \cdot | 0 \rangle)$-QSP convention, our given signal unitary $W_{Z}$ is defined as

\[
W_{Z} \left(\theta \right) \equiv e^{i \frac{\theta}{2} Z} =
  \begin{bmatrix}
    e^{i \frac{\theta}{2}} & 0 \\
    0 & e^{-i \frac{\theta}{2}}
  \end{bmatrix}
  = \begin{bmatrix}
    \omega & 0 \\
    0 & \omega^{*}
  \end{bmatrix}  = U 
\]
where $\omega \equiv e^{i \frac{\theta}{2}}$ and $Z$ denotes the Pauli-Z operator. The associated signal processing rotation operator is $S_{X} \left( \phi \right) = e^{i \phi X}$ and the signal basis for this convention is  $\left\lbrace | 0 \rangle , |1 \rangle \right\rbrace$. The QSP operation sequence here is defined as
\begin{equation*}
U_{\vec{\phi}}= e^{i \phi_{0} X} \prod_{k=1}^{d} W_{Z} \left( \theta \right) e^{i \phi_{k} X}
\end{equation*}
for the tuple of phase angles $\vec{\phi} = \left(\phi_{0}, \phi_{1}, \dots \phi_{d} \right) \in \mathbb{R}^{d+1}$.
\item 
Under the $(W_{X}, S_{Z}, \langle + | \cdot | + \rangle)$-QSP convention, our given signal unitary $W_{X}$ is defined as

\[
W_{X} \left(\theta \right)  \equiv e^{i \frac{\theta}{2} X} =
  \begin{bmatrix}
    a & i \sqrt{1 - a^{2}} \\
    i \sqrt{1 - a^{2}} & a
  \end{bmatrix}
\]
where $a \equiv cos(\frac{\theta}{2})$ and $X$ denotes the Pauli-X operator. The signal processing rotation operator is $S_{Z} \left( \phi \right) = e^{i \phi Z}$ and the signal basis is $\left\lbrace | + \rangle , |- \rangle \right\rbrace$. As before, the QSP operation sequence here is similarly defined as
\begin{equation*}
U_{\vec{\phi}}= e^{i \phi_{0} Z} \prod_{k=1}^{d} W_{X} \left( \theta \right) e^{i \phi_{k} Z}.
\end{equation*}
\end{enumerate}
\end{definition}
We note that both these conventions are equally expressive (in terms of the achievable polynomials) as
\begin{equation*}
\langle 0 | e^{i \phi_{0} X} \prod_{k=1}^{d} W_{Z} \left( \theta \right) e^{i \phi_{k} X} | 0 \rangle = \langle + | e^{i \phi_{0} Z} \prod_{k=1}^{d} W_{X} \left( \theta \right) e^{i \phi_{k} Z} | + \rangle 
\end{equation*}
since for the Hadamard gate $H = \frac{1}{\sqrt{2}} \begin{psmallmatrix} 1 & 1\\ 1 & -1\end{psmallmatrix}$, we have $H| 0 \rangle = |+ \rangle$, $HZH = X$, $HXH = Z$, and $H^{2} = I$.

Given a rotation operator $W_{X}(\theta)$,  for notational convenience (and a slight abuse of notation) we introduce $a = \text{cos}(\frac{\theta}{2})$ as an argument for $W_{X}$ such that
\be 
W_{X}(\theta) := e^{i \frac{\theta}{2} X}= \mqty[a & i\sqrt{1-a^2} \\ i\sqrt{1-a^2} & a] = W_{X}(a)
\ee
and the QSP sequence $U_{\vec{\phi}}$ produces a matrix which may be expressed in terms of
polynomials $P(a)$ and $Q(a)$ as
\be 
e^{i\phi_{0}Z}\prod_{k=1}^{d}W_{X}(a)e^{i\phi_{k}Z} = \mqty[ P(a) & iQ(a)\sqrt{1-a^2} \\  iQ(a)\sqrt{1-a^2} & P(a)]
\ee 
such that:
\begin{enumerate}
    \item $deg(P) \leq d,\; deg(Q) \leq d-1$,
    \item $P$ has parity $d$ mod $2$ and $Q$ has parity $(d-1)$ mod $2$,
    \item $\left|P \right|^{2} + \left(1 - a^{2} \right) \left| Q \right|^{2} = 1$.
\end{enumerate}
We now state the key QSP result necessary for this work that is sufficiently expressive in terms of the achievable polynomials under both the $(W_{Z}, S_{X}, \langle 0 | \cdot | 0 \rangle)$ and $(W_{X}, S_{Z}, \langle + | \cdot | + \rangle)$-QSP conventions.

\begin{theorem}[QSP Theorem] 
\label{thm:qsp}
There exists a QSP phase sequence $\vec{\phi} = \left(\phi_{0}, \phi_{1}, \dots \phi_{d} \right) \in \mathbb{R}^{d+1}$ such that
\begin{equation}
{\rm Poly}(a) = \langle + | e^{i \phi_{0} Z} \prod_{k=1}^{d} W_{X} \left( a \right) e^{i \phi_{k} Z} | + \rangle = \langle 0 | e^{i \phi_{0} X} \prod_{k=1}^{d} W_{Z} \left( \theta \right) e^{i \phi_{k} X} | 0 \rangle 
\end{equation}
for $a \in \left[ -1,1 \right]$, and for real polynomial Poly $\in \mathbb{R}\left[a \right]$ if and only if:
\begin{enumerate}
    \item  deg(Poly) $\leq d$
    \item \text{Poly} has parity $d$ mod $2$
    \item $\forall a \in \left[-1,1 \right], \left|Poly(a) \right| \leq 1 $.
\end{enumerate}
\end{theorem}
Recently, a new form of QSP theorem has been found that removes the requirement that the polynomial has a fixed parity.  However, this comes at the price of requiring two parameterized rotations rather than one per step as well as requiring a slightly different method for finding the phase angles for the QSP.  Our work will use both of these approaches.  We state a form of this generalized QSP theorem below.

\begin{theorem}[Generalized QSP Theorem]\label{thm:GQSP}
    $\forall d \in \mathbb{N}, \exists\; \vec{\theta}, \vec{\omega}\in \mathbb{R}^{d+1}, \lambda\in \mathbb{R}\; such \; that$ 
    \begin{equation}
        U_{\vec{\theta}, \vec{\omega}} = \left( \prod_{j=1}^{d} R(\theta_j, \omega_j,0) (\dyad{0} \otimes U + \dyad{1} \otimes I) \right) R(\theta_0, \omega_0,\lambda) = \mqty[P(U) & . \\ Q(U) & .]
        \label{eq:gqsp}
    \end{equation}
    \begin{enumerate}
    \item $P,Q:\mathbb{T}\mapsto \mathbb{C}$ and $deg(P),\; deg(Q) \le d$
    \item $ \abs{P(a)}^2 + \abs{Q(a)}^2 = 1$, $\forall a \in \mathbb{T}\;$, where $\mathbb{T}=\{ a\in \mathbb{C}:\abs{a}=1\}$ 
    \item $R(\theta, \omega,\lambda) = \mqty[e^{i(\lambda + \omega)}\cos(\theta) & e^{i\omega \sin(\theta)}\\ e^{i\lambda \sin(\theta)} & -\cos(\theta)] \otimes I$.
\end{enumerate}
\end{theorem}
By interleaving the above $R(\theta, \omega,0)$ and $W$- operations we can block encode polynomial transformations of $a$ without further assumptions. We only need to specify a single $\lambda$ since we can absorb the other instances into the instance of $\omega$ before them.

Logarithmic block-encoding of a given signal unitary plays a key role in our method of iterated QSP. We now state a key lemma that will be repeatedly utilized in our later constructions.
\begin{lemma}[Logarithmic block-encoding]
\label{lem:log}
If the spectral norm of an operator $ \mathcal{H} $ is at most $\frac{1}{2}$, given a signal unitary $U = e^{i \mathcal{H}}$, and for $\epsilon \in \left(0, \frac{1}{2} \right]$, a $ \left( \frac{\pi}{2}, 2 ,\epsilon \right)$-block-encoding of $\mathcal{H}$ can be implemented with $\mathcal{O}\left(\log(\frac{1}{\epsilon})\right)$ uses of controlled-$U$ and its inverse, with $\mathcal{O}\left(\log(\frac{1}{\epsilon})\right)$ two-qubit gates and a single ancilla qubit.
\end{lemma}
\begin{proof}
Proof follows the discussion in~\cite{Gily_n_2019}, reproduced here with slight modifications for completeness.
First we observe that
\begin{equation}
 \sin \left( \mathcal{H} \right) = -i \left( \langle + | \otimes I \right) cU^{\dagger} \left( ZX \otimes I \right) cU \left( | + \rangle \otimes I \right) \label{eq:polyarcsin}
\end{equation}
where $cU$ denotes the controlled-$U$ gate controlled by the first qubit (and $cU^{\dagger}$ similarly that corresponding to $U^{\dagger}$), $X$ denotes the Pauli-X operator and $Z$ denotes the Pauli-Z operator. The quantum circuit associated with \ref{eq:polyarcsin} gives a two-qubit block-encoding of $\sin(\mathcal{H})$ (Corollary 71 of \cite{Gily_n_2019}). Using Lemma 70 of \cite{Gily_n_2019}, if $\delta, \epsilon \in \left(0, \frac{1}{2} \right]$, there exists an efficiently computable odd polynomial $P \in \mathbb{C}[x]$ with real part that corresponds to the real polynomial $P_{\text{R}} \in \mathbb{R}[x]$ of degree $\mathcal{O}\left( \frac{1}{\delta} \log(\frac{1}{\epsilon})\right)$ such that $||P_{R} ||_{[-1,1]} \leq 1$ and
\begin{equation}
\left\Vert P_{\text{R}}(x) - \frac{2}{\pi} {\arcsin(x)} \right\Vert_{[-1+ \delta, 1 - \delta]} \leq \epsilon \label{eq:polyarcsin_bound}
\end{equation}
since for all $x \in \left[-1,1 \right]$ we have $\frac{2}{\pi} \text{arcsin(x)} = \sum_{l=0}^{\infty} {2l \choose l} \frac{2^{-2l}}{2l+1}\frac{2}{\pi}x^{2l+1} $, which implies that $\sum_{l=0}^{\infty} \left\| {2l \choose l} \frac{2^{-2l}}{2l+1}\frac{2}{\pi} \right\| = 1$, thus satisfying the polynomial boundedness condition for Corollary 66 of \cite{Gily_n_2019}. Considering only the real part of the computable odd polynomial $P$ (corresponding to an approximation of $\frac{2}{\pi} \text{arcsin}(x)$), we then utilize QSP to generate an approximation of $\mathcal{H}$ satisfying \ref{eq:polyarcsin_bound}. Taking $\delta = \frac{1}{2}$ above, we thus obtain an $\epsilon$-approximating polynomial $P$ of $\frac{2}{\pi} \text{arcsin}(x)$ over $\left[-\frac{1}{2}, \frac{1}{2} \right]$, thereby achieving a $ \left( \frac{\pi}{2}, 2 ,\epsilon \right)$-block-encoding of $\mathcal{H}$ as claimed.
\end{proof}

\section{Iterated ${\rm SU}(2)$ QSP} \label{sec:iterQSP}
The central idea behind our work is simple.  What we aim to do is to apply a form of modular QSP (as coined in~\cite{rossi2023modular}) in an iterative process to construct our target polynomial.
Specifically, if we have a signal unitary in the $W_X$ basis of $\{e^{i \frac{\theta_i}{2} X}: i\in \{1,\ldots,R\}\}$ and a signal processing rotation operator of the form $S_Z = e^{i \phi Z}$, we would like to ask what unitaries are achievable via QSP using these operations.  The results in this section are stated in the $(W_X,S_Z, \bra{+}\cdot \ket{+})$-QSP convention for concreteness but can be trivially converted to the $(W_Z,S_X, \bra{0}\cdot \ket{0})$-QSP convention as needed via Definition \ref{def:qsp conventions} above through the appropriate use of Hadamard gates.

\begin{theorem}[Iterated ${\rm SU}(2)$ QSP]\label{thm:iteratedQSP}
    Let $\{\theta_1,\ldots, \theta_R\}$ be a set of angles and let $W_X(\theta_j)$ be a set of achieveable signal matrices in the $(W_X,S_Z, \bra{+}\cdot \ket{+})$-QSP convention using a sequence of length $(d_1,\dots,d_{R})$ of fundamental signal or $S_Z$ signal processing rotations.  If $q_k$ is a sequence over the indices $\{1,\ldots, k\}$ and if $P:\mathbb{T}\mapsto \mathbb{C}$ where $\mathbb{T}=\{ a\in \mathbb{C}:\abs{a}=1\}$ is a polynomial function of degree $d$ that satisfies the requirements for the $P$ function in Theorem~\ref{thm:GQSP} then the operation 
    $$
W_X\left(2\arccos\left(P\left(\cos\left(\sum_{s\in q_k} \frac{\theta_s}{2} \right)\right)\right)\right)
    $$
can be constructed using a number of fundamental signal or $W_X$ rotations that is in
$
O\left(d\left( \sum_{s\in q_k} d_s\right) \right).
$
\end{theorem}
\begin{proof}
There are two steps needed in this construction.  First note that for all $j$ and $j'$ we have that
\begin{eqnarray}
    W_X(\theta_j) W_X(\theta_{j'}) = W_X(\theta_j + \theta_{j'}).
\end{eqnarray}
The number of rotations needed to construct each such rotation using QSP or other rotation synthesis method is assumed to be $d_j$ and $d_{j'}$ respectively.  Our claim is that if we perform a sum of such terms the cost is simply the sum of the individual costs.  Thus we can use the above observation in a base case that inductively proves our cost assumption.  Specifically, assume that for some sequence $q_{k'}$ of length $R$ that this hypothesis is true.  Then let $q_{k''} = \{q_{k'},v\}$ for $v \in \{1,\ldots,R\}$. we have from our previous observation that
\begin{equation}
\prod_{s\in q_{k''}}W_X(\theta_s) = \prod_{s\in q_{k'}}W_X(\theta_s) W_X(\theta_v)= W_X(\sum_{s\in q_{k''}} \theta_s).
\end{equation}
Thus the result holds by induction for all length strings.

Then by definition if we have such a rotation then the upper-left matrix element of $W_X(\sum_{s\in q_{k''}} \theta_s)$ is
\begin{equation}
    a = \cos(\sum_{s\in q_{k}} \theta_s/2).\label{eq:a}
\end{equation}
From the generalized QSP theorem~\cite{motlagh2023generalized}, stated as Theorem~\ref{thm:GQSP} above, we then have that for any polynomial transformation satisfying the assumptions of Thoerem~\ref{thm:GQSP} and $Q=i|Q|$ such that $Q$ also satisfies the assumptions of Theorem~\ref{thm:GQSP}, we have a resulting transformation that is also a $R_X$ rotation as required.  Thus such a transformation is achievable conditioned on $P$ and $Q$ satisfying the GQSP assumptions.  Further the angle generated is $\phi = 2\arccos(P(\cos(\sum_{s\in q_{k}} \theta_s/2)))$ from Theorem~\ref{thm:GQSP}
and~\eqref{eq:a}.
As the degree of the polynomial is at most $d$, the total number of rotations needed is given by Theorem~\ref{thm:GQSP} to be
\begin{equation}
    (d+1) +d(\sum_{s\in q_k} d_s) \in O\left(d \sum_{s\in q_k} d_s \right).
\end{equation}
as claimed.
\end{proof}

\section{Application: multivariate-QSP (M-QSP)}\label{sec:MQSP}
In this section, we describe our multivariate-QSP (M-QSP) approach in detail. Our approach differs from earlier attempts \cite{Rossi_2022, rossi2023modular} to create a multivariate version of QSP as we do not rely on the correctness of (and inherit the attendant complications from) the strong conjecture that unitary matrix completions (with minimal symmetries and constraints) provided
by the multivariate Fej\'er-Riesz theorem (FRT) can always be constructed as low-degree iterated products of unitaries. That is, we do not require that the multivariate FRT = QSP conjecture be true. Recent work \cite{németh2023variants, mori2023comment} has provided counterexamples to the multivariate FRT = QSP conjecture hence necessitating that additional conditions are required for earlier versions of M-QSP to work since the multivariate FRT = QSP conjecture by itself no longer suffices.
The central idea behind our version of M-QSP is to iteratively use QSP to provide analytically tractable methods to construct more complicated (multivariate) functions of block encoded matrices under the $(W_{Z}, S_{X}, \langle 0 | \cdot | 0 \rangle)$-QSP convention. An essential component of our approach, involves an algorithm that obliviously transforms a given signal unitary of the form $W_{Z} \left(\theta \right) \equiv e^{i \frac{\theta}{2} Z}$ into $e^{i (\frac{\theta}{2})^{2} Z}$. To this end, our general strategy will be to first determine an appropriate approximating polynomial for each desired transformation and then apply appropriate QSP phases (that generate this approximating polynomial under QSP) to a series of signal processing rotation operators interleaved alternately with an appropriate (block-encoded) input signal unitary. Our algorithm that obliviously transforms our given signal unitary $W_{Z} \left(\theta \right) \equiv e^{i \frac{\theta}{2} Z}$ into $e^{i (\frac{\theta}{2})^{2} Z}$ can then be used to generate arbitrary multivariate polynomial functions of phases as required for M-QSP.

Given an input signal unitary of the form $W_{Z} \left(\theta \right) \equiv e^{i \frac{\theta}{2} Z} = e^{i \mathcal{H}} $, where $\mathcal{H} = \frac{\theta}{2} Z$, we first utilize logarithmic block-encoding via Lemma \ref{lem:log} to obtain $\frac{\theta}{2} Z$. We note that the circuit corresponding to \eqref{eq:polyarcsin} is equivalent to implementing
\begin{equation}
\sin \left( \frac{\theta}{2} Z \right) =  \left( \langle 0 | \otimes I \right) \left( \text{H} \otimes I \right) cU^{\dagger} \left( Y \otimes I \right) cU \left( \text{H} \otimes I \right) \left( | 0 \rangle \otimes I \right) \label{sinethetaZ2}
\end{equation}
where $Y$ denotes the Pauli-Y operator and $U = e^{i \mathcal{H}}$. Thus, $ \text{sin}\left( \frac{\theta}{2} Z \right)$ can be block-encoded in a  two qubit quantum circuit corresponding to the operator product $U_{A} =\left( \text{H} \otimes I \right) cU^{\dagger} \left( Y \otimes I \right) cU \left( \text{H} \otimes I \right)$. 

We now employ the quantum eigenvalue transformation (QET) for real polynomials 
applied to our block-encoding of $ \text{sin}\left( \frac{\theta}{2} Z \right)$ to obtain a logarithmic block-encoding of $\frac{\theta}{2}Z$. By the principle of deferred measurement, we may treat $\text{sin}\left( \frac{\theta}{2} Z \right)$ as already post-selected, keeping in mind that this is predicated on a successful projective measurement on the $ | 0 \rangle \langle 0|$ index qubit as $\text{sin}\left( \frac{\theta}{2} Z \right) = \left( \langle 0 | \otimes I \right) U_{A} \left( | 0 \rangle \otimes I\right) $.
Consider an orthogonal projector $\Pi = | 0 \rangle \langle 0 |$, the unitary $U_{A}$ and $\Pi$ form a projected unitary encoding of the operator $\text{sin}\left( \frac{\theta}{2} Z \right)$ if 
\begin{equation}
\Pi  U_{A} \Pi =  \text{sin} \left( \frac{\theta}{2} Z \right).
\end{equation} 
Equivalently, we also have that 
\begin{equation}
\left( \langle 0 | \otimes I \right) U_{A} \left( | 0 \rangle \otimes I \right) = \text{sin} \left( \frac{\theta}{2} Z \right).
\end{equation} 


\begin{lemma}[Quantum eigenvalue transformation for odd polynomials] Given a block-encoding of a Hamiltonian $\mathcal{H} = \sum_{\lambda} \lambda | \lambda \rangle \langle \lambda |$ in a unitary matrix $U$ such that
\begin{equation}
\mathcal{H} = \Pi U \Pi
\end{equation}
with projector $\Pi$ locating $\mathcal{H}$ within $U$, and access to projector controlled phase shift operations $\Pi_{\phi_{j}} = e^{i \phi_{j}(2 \Pi - I)}$, then for odd $d$,
\begin{equation}
U_{\vec{\phi}} = e^{i \phi_{1} (2 \Pi - I )} U \prod_{j=1}^{(d-1)/2} \left( e^{i \phi_{2j} (2 \Pi - I )} U^{\dagger} e^{i \phi_{2j+1} (2 \Pi - I )} U \right) = \begin{bmatrix}
    {\rm Poly}\left(\mathcal{H} \right) & \cdot \\
    \cdot & \cdot
  \end{bmatrix} \label{UPhi}
\end{equation}
with
\begin{equation}
{\rm Poly}(\mathcal{H}) = \sum_{\lambda} {\rm Poly} \left(\lambda \right) | \lambda \rangle \langle \lambda |
\end{equation}
a polynomial transformation of the eigenvalues of $\mathcal{H}$. This polynomial is of degree at most $d$ and observes the conditions on polynomial $P$ from Theorem \ref{thm:qsp}.
\end{lemma}

\begin{lemma}\label{lem:square}
Assume that we have a signal unitary $W_Z(\theta) = e^{i \frac{\theta}{2} Z} \in \mathbb{C}^{2\times 2}$ and let $\epsilon' >0$ be an error tolerance.  We can then construct a $(1+\epsilon',6,\epsilon')$-block encoding of $\exp(i(\frac{\theta}{2})^2 Z)$ using a number of queries to $W_Z(\theta)$ and additional single qubit rotations that scales as
$$
\mathcal{O}\left(\frac{\log^2(1/\epsilon')}{\log\log(1/\epsilon')} \right)
$$
\end{lemma}
\begin{proof}
By interleaving $U_{A}$ and $U_{A}^{\dagger}$ with projector controlled phase shift operations $\Pi_{\phi_{j}} = e^{i \phi_{j}(2 \Pi - I)}$ corresponding to a vector of phases $\vec{\phi} \in \mathbb{R}^{d+1} $ (where d is ${\mathcal{O}\left( \log(\frac{1}{\epsilon})\right)}$), we thus see that QET performs an oblivious polynomial transformation $P$ corresponding to $\frac{2}{\pi}\arcsin(x)$ of the eigenvalues of $U_{A}$ via the unitary $U_{\vec{\phi}}$  defined for odd polynomials of degree $d$. Phase angles generating a particular $P$ exist and can be efficiently calculated on a classical computer for our approximating polynomial $P$ in time $ \mathcal{O} ({\rm poly}(d, \log(1/ \epsilon)))$, if the degree of $P$ and the associated polynomial $Q$ (related to a matrix completion given $P$) are less than $d$ and $d-1$ respectively (in addition to satisfying the corresponding parity and unitarity constraints).

By considering a tuple of negative phases $ - \vec{\phi}$, QET then gives (in analogy with Corollary 18 in \cite{Gily_n_2019}) that
\begin{eqnarray}
P^{(QET)} \left( \Pi  U_{A}  \Pi \right) /2  &=& \left( \langle + | \otimes \Pi \right) \left( | 0 \rangle \langle 0 | \otimes U_{\vec{\phi}} \right) \left( | + \rangle \otimes \Pi \right) \\
P^{*(QET)} \left( \Pi  U_{A}  \Pi \right) /2  &=& \left( \langle + | \otimes \Pi \right) \left( | 1 \rangle \langle 1 | \otimes U_{-\vec{\phi}} \right) \left( | + \rangle \otimes \Pi \right)
\end{eqnarray}
where from (\ref{UPhi}) we observe that a vector of negative phases $-\vec{\phi}$, gives the original polynomial $P^{(QET)}$ instead with complex conjugated coefficients as $P^{*(QET)}$. Combining these gives the real polynomial $P_{R}^{(QET)}$ where
\begin{eqnarray}
P_{R}^{(QET)} \left( \Pi  U_{A}  \Pi \right) &=& P^{(QET)} \left( \Pi  U_{A} \Pi \right) /2 + P^{*(QET)} \left( \Pi  U_{A} \Pi \right) /2 \\ &=& ( \langle + | \otimes \Pi ) \left( |0 \rangle \langle 0| \otimes U_{\vec{\phi}} + |1 \rangle \langle 1| \otimes U_{-\vec{\phi}}\right) \left( | + \rangle \otimes \Pi \right).
\end{eqnarray}
Putting this all together, for the polynomial $P$ (which approximates $\frac{2}{\pi} \text{arcsin(x)}$), we have from Theorem 17 in \cite{Gily_n_2019} that 
\begin{eqnarray}
P^{(QET)} \left( \Pi U_{A} \Pi \right) &=& \Pi U_{\vec{\phi}} \Pi \\
P^{(*QET)} \left( \Pi U_{A} \Pi \right) &=& \Pi U_{-\vec{\phi}} \Pi
\end{eqnarray}
and thus
\begin{equation}
\norm\bigg{ \frac{\theta}{2}Z - \frac{\pi}{2} \left( \langle 0 |^{\otimes 2} \otimes I \right) \left( |0 \rangle \langle 0| \otimes U_{\vec{\phi}} + |1 \rangle \langle 1| \otimes U_{-\vec{\phi}} \right) \left( | 0 \rangle^{\otimes 2} \otimes I\right) } \leq \epsilon
\end{equation}
such that $U_{\Phi}$ is a $ \left( \frac{\pi}{2}, 2 ,\epsilon \right)$-logarithmic block-encoding of $\frac{\theta}{2}Z$ as desired. The QET step requires $\mathcal{O} (\log(\frac{1}{\epsilon}))$ queries to the block-encoding of $\sin \left(\frac{\theta}{2} Z \right)$ and this number of queries corresponds to the degree of the approximating polynomial $P$ for an error tolerance of $\epsilon$, thus confirming the $\mathcal{O} (\log(\frac{1}{\epsilon}))$ uses of controlled-$U$ and its inverse.

We next perform phase squaring and exponentiation to obtain ${e^{i(\frac{\theta}{2})^{2} Z}}$ via the following steps.  
Using $U_{\vec{\phi}}$, our $ \left( \frac{\pi}{2}, 2 ,\epsilon \right) $-logarithmic block-encoding of $\frac{\theta}{2} Z $, we then see that  $(I \otimes U_{\vec{\phi}}) ( I \otimes U_{\vec{\phi}} )$ gives a $\left( (\frac{\pi}{2})^{2}, 4, \pi \epsilon \right) $-block-encoding of $ ( \frac{\theta}{2}Z )^2$. In order to ensure that this transformed block-encoded operator does not act trivially (in projection) since $Z^{2} = I$, we 'multiply' an additional block-encoded Pauli-Z operator, so we require the following product of block-encodings $U' = ( I \otimes U_{\vec{\phi}}) ( I \otimes U_{\vec{\phi}} ) ( I \otimes I \otimes Z )$ which thus gives a $\left( (\frac{\pi}{2})^{2}, 6, \pi \epsilon \right)$-block-encoding of $(\frac{\theta}{2})^{2} Z$ where
\begin{equation}
\norm\bigg{ \left( \frac{\theta}{2} \right)^{2}Z - \left( \frac{\pi}{2} \right)^{2} \left( \langle 0|^{\otimes 2} \otimes I \right) U_{\vec{\phi}} \left( |0 \rangle ^{\otimes 2} \otimes I \right) \left( \langle 0|^{\otimes 2} \otimes I \right) U_{\vec{\phi}} \left( |0 \rangle ^{\otimes 2} \otimes I \right) \left( \langle 0|^{\otimes 2} \otimes I \right) I \otimes I \otimes Z \left( |0 \rangle^{\otimes 2} \otimes I \right) } \leq \pi \epsilon.
\end{equation}

By setting $\epsilon' = \epsilon/ \pi$, we thus obtain a $\left( (\frac{\pi}{2})^{2}, 6, \epsilon' \right)$-block-encoding of $(\frac{\theta}{2})^{2} Z$.  Using this block-encoding of $(\frac{\theta}{2})^{2} Z$, we now again utilize QET to obtain $e^{i(\frac{\theta}{2})^{2} Z} $. Since the exponential function $e^{ix}$ does not have a definite parity, we will need to apply GQSP to produce
\begin{equation}
\cos^{(QET)} \left( \left(\frac{\theta}{2} \right)^{2}Z  \right) + i \sin^{(QET)}\left( \left(\frac{\theta}{2}\right)^{2} Z \right) =  e^{i(\frac{\theta}{2})^{2} Z} \label{cos_sine_sum}
\end{equation}
so as to correctly generate our target $e^{i(\frac{\theta}{2})^{2} Z}$.  This requires an additional ancillary qubit to implement the polynomial transformation and we still need to obtain the required approximating polynomial transformations.

QET approximating polynomials for cosine and sine can be obtained using the Jacobi-Anger expansions for $\cos(xt)$ and $\sin(xt)$
\begin{eqnarray}
\cos(xt) &=& \textit{J}_{0}(t) + 2 \sum_{k=1}^{\infty} (-1)^{k} \textit{J}_{2k}(t) \textit{T}_{2k}(x) \\
\sin(xt) &=& 2 \sum_{k=1}^{\infty} (-1)^{k} \textit{J}_{2k+1}(t) \textit{T}_{2k+1}(x)
\end{eqnarray}
where $\textit{J}_{i}(t)$ is the $i$-th order Bessel function of the first kind and $\textit{T}_{i}$ is the $i$-th order Chebychev polynomial of the first kind. From Hamiltonian simulation by QSP \cite{Gily_n_2019} \cite{Low_2019} \cite{Low_2017}, we know that to obtain $\epsilon'$-approximations for $cos(xt)$ and $sin(xt)$ using these expansions, we will need to truncate these expansions at sufficiently large degrees of $2R$ and $2R+1$ respectively, where $R = \left\lfloor \frac{1}{2} r \left( \frac{e|t|}{2}, \frac{5}{4} \epsilon' \right)\right\rfloor $, $t \in \mathbb{R} \backslash \{0\}$ and $\epsilon' \in (0, \frac{1}{e})$. The function $r(t, \epsilon')$ is defined implicitly via
\begin{equation}
\epsilon = \left( \frac{t}{r} \right)^{r}
\end{equation}
and scales asymptotically as
\begin{equation}
r(t, \epsilon') = \Theta \left( |t| + \frac{\log (1/\epsilon')}{\log \left(e + \frac{\log(1/\epsilon')}{|t|} \right)} \right).
\end{equation}
We can thus reframe our approach to obtaining a block-encoding of $e^{i(\frac{\theta}{2})^{2}Z} $ as a Hamiltonian simulation of $\mathcal{H}$ for a time of $t$, where $\mathcal{H} =   (\frac{\theta}{2})^{2} Z $ and $t=-1$. To use QSP, the even and odd approximating polynomials $P^{\cos}_{\epsilon'} (x,t)$ and $P^{\sin}_{\epsilon'} (x,t)$ need to be re-scaled so as to be correctly bounded by $1$ in absolute value. To ensure an $\epsilon'$-approximation to $e^{ixt}$, we re-scale both these polynomials by $\frac{1}{1+ \epsilon' / 4}$ thus giving $ \left| \frac{1}{1+ \epsilon' / 4} P^{cos}_{\epsilon'} (x,t) - \cos(xt) \right| \leq \frac{\epsilon'}{2} $ and $ \left| \frac{1}{1+ \epsilon' / 4} P^{\sin}_{\epsilon'} (x,t) - \sin(xt) \right| \leq \frac{\epsilon'}{2} $ as required. QSP in this exponentiation step requires $2R$ queries for the even polynomial and $2R +1$ queries for the odd polynomial to our modified $U'$ (our block-encoding of the appropriately scaled $\mathcal{H'}$) giving a total query count of
\begin{equation}
2R +2R + 1 = \Theta \left( 1+ \frac{\log (1/\epsilon')}{\log \left( e + \frac{1}{2}\log(1/\epsilon') \right)} \right).
\end{equation}
This scaling has been demonstrated for Hamiltonian simulation \cite{Low_2017} \cite{Gily_n_2019} to be optimal in $\epsilon'$ and is similarly optimal for our construction here. As before, the required QSP phase angles can be efficiently calculated for our approximating polynomials $P^{\cos}_{\epsilon'} (x,t)$ and $P^{\sin}_{\epsilon'} (x,t)$ on a classical computer in time $ \mathcal{O} (\text{poly}(n, \log(1/ \epsilon')))$ where $n$ is the degree of the approximating polynomial. Applying these phase angles via QSP to obtain $P^{\cos}_{\epsilon'} ((\frac{\theta}{2})^{2} Z)$ and $P^{\sin}_{\epsilon'} ((\frac{\theta}{2})^{2} Z)$, these two parts can be combined after a re-scaling using a linear combination of their unitary block-encodings to give an $\epsilon'$-approximation of our target $e^{i(\frac{\theta}{2})^{2} Z}$ as required.  

The result of Theorem~\ref{thm:iteratedQSP} shows that the iterated block-encodings of the two processes is the product of the two and so the overall cost is $\mathcal{O}(\log^2(1/\epsilon')/\log\log(1/\epsilon'))$ as claimed.
\end{proof}
This shows that we can construct a phase multiplier using a constant number of ancillary qubits.  Further the cost of the multiplier is actually competitive with the elementary multiplication algorithm which scales as $\mathcal{O}(\log^2(1/\epsilon))$ but requires more than a constant number of qubits and does not apply when the input is provided as phases.
Our goal of M-QSP is to generate multivariate polynomial transformations of the eigenphases of a given set of (univariate) signal unitaries. To achieve these multivariate polynomial transformations, we combine our above phase squaring algorithm with iterated use of QSP. The central idea behind our work is to iteratively use QSP to provide analytically tractable methods to construct complicated functions of block encoded matrices via the use of suitable approximating polynomials.

Suppose we are given oracle access to a set of $m$ signal unitaries $\lbrace U_{k} = e^{i \frac{\theta_{k}}{2} Z} | k = 1,2, \dots , m \rbrace $ and their inverses. As before, we are working in the $(W_{Z}, S_{X}, \langle 0 | \cdot | 0 \rangle)$-QSP convention. Our goal is to construct a unitary with a phase corresponding to a multivariate polynomial function of the signal unitaries' phases.  
A general multivariate polynomial of univariate phases $ P(\theta_{1}, \theta_{2}, \dots, \theta_m{}) $ has two distinct types of terms:
\begin{enumerate}
\item powers of univariate phases $(\frac{\theta_{k}}{2})^{l}$
\item powers of multivariate phases $\Pi_{k=1}^{m} (\frac{\theta_{k}}{2})^{l_{k}}$.
\end{enumerate}
We next discuss how to construct each of these in turn.

\begin{corollary}\label{cor:PUP}
For any $l'\in \mathbb{Z}_+$ we can construct a $(1+\epsilon',2l', \epsilon' )$-block encoding of $\exp(i(\frac{\theta}{2})^{l'} Z)$ using
a number of applications of the $W_Z(\theta)$ and single qubit unitaries that scale as
$$
\mathcal{O}\left(\frac{(\frac{\pi}{2})^{l'}\log^2(l'/\epsilon')}{\log\log(l'/\epsilon')} \right)
$$
\end{corollary}
\begin{proof}
For powers of univariate phases, we can immediately adapt our phase squaring and exponentiation algorithm detailed previously to generate any univariate power we require. In this instance, there are two different cases, even and odd powers of the univariate phase. 

For the case of even powers $l$  of a univariate phase, $(\frac{\theta_{k}}{2})^{l}$, using $U_{\vec{\phi},k}$, from Lemma \ref{lem:log} and our earlier discussion of logarithmic block-encoding, the $ \left( \frac{\pi}{2}, 2 ,\epsilon \right) $-logarithmic block-encoding of $\frac{\theta_{k}}{2} Z $, $(I \otimes U_{\vec{\phi},k})^{l}$ gives a $\left( (\frac{\pi}{2})^{l}, 2l, \frac{\pi \epsilon}{2} l \right) $-block-encoding of $ ( \frac{\theta_{k}}{2}Z )^l$. In order to ensure that this transformed block-encoded operator does not act trivially (in projection) since $Z^{2} = I$, we `multiply' an additional block-encoded Pauli-Z operator, so $U' = ( I \otimes U_{\vec{\phi},k})^{l} ( I^{\otimes l} \otimes Z )$. This gives a $\left( (\frac{\pi}{2})^{l}, 2l+1, \frac{\pi \epsilon}{2} l \right) $-block-encoding of $(\frac{\theta}{2})^{l} Z$, which we can exponentiate to $e^{i(\frac{\theta}{2})^{l} Z}$ following the remainder of the phase squaring algorithm as before. Note that any arbitrary coefficient $a \in \mathbb{R}$ can also be appended to the phase by setting $t = a$  in the phase exponentiation step to generate $e^{i \: a(\frac{\theta}{2})^{l} Z}$ with the correct phase coefficient as required. For an odd univariate power $l'$, we proceed as before but we do not need to multiply an extra $( I^{\otimes l'} \otimes Z )$ because $Z^{l'} = Z$ for odd $l'$. The rest of the algorithm remains identical and we obtain a $\left( (\frac{\pi}{2})^{l'}, 2l, \frac{\pi \epsilon}{2} l' \right) $-block-encoding of $(\frac{\theta}{2})^{l'} Z$, which we can then exponentiate to $e^{i(\frac{\theta}{2})^{l'} Z}$ via the exponentiation step of the phase squaring algorithm as before.  This leads to a total cost of $\mathcal{O}((\pi/2)^{l'}\log^2(\epsilon' l')/\log\log(\epsilon' l'))$.
\end{proof}

\subsection{Squaring construction for small $\theta$}
The previous construction for squaring circuits used a conventional approach to block encode the rotation angle of an input $Z$-rotation.  This construction works well for large values of $\theta$, but the block-encoding constant involved makes it less efficient for small inputs.  In addition, the construction of the block encoding leads to constant overheads in the number of ancillae.  Here we provide a different approach that uses a polynomial approximation to block-encode the rotation angle and ultimately square it that obviates these difficulties at the price of restricting the input rotations such that $|\theta|\le 1/3$.  

\begin{lemma}\label{lem:JacobiAngerSq} For any $a\in \mathbb{R}$ we have that
    $$
e^{ia^2}=J_0(a)+2\sum_{n> 0}^\infty (i)^n J_n(a)T_n(a)
    $$
and for any $\epsilon>0$ there exists $M^*\in \mathbb{Z}$ with $M^* \in \mathcal{O}(\log(1/\epsilon)/\log\log(1/\epsilon))$ such that for any integer $M>M^*$ 
$$
|e^{ia^2} - J_0(a)+2\sum_{n> 0}^M (i)^n J_n(a)T_n(a)| \le \epsilon
$$
\end{lemma}
\begin{proof}
From the Jacobi Anger expansion we have that for any $\theta,a\in \mathbb{R}$
\begin{equation}
    e^{ia\cos(\theta)} = J_0(a) + 2\sum_{n>0}^\infty (i)^n J_n(a) \cos(n\theta).
\end{equation}
Next we choose $\theta = \arccos(a)$ and the fact that $T_n(a) := \cos(n\arccos(a))$ to see the claimed result.

To see the convergence result, note that this is a special case of the Jacobi-Anger expansion with $\theta = \arccos(a)$.  Previous work has shown through the asymptotics of Bessel functions that for any $\theta$~\cite{Gily_n_2019,Low_2017}
\begin{equation}
    2\sum_{n>M^*}  |J_n(a)||\cos(\theta)|\le 2\sum_{n>M^*}  |J_n(a)| \le \epsilon
\end{equation}
if
\begin{equation}
    M^* \in \mathcal{O}\left(\frac{\log(1/\epsilon)}{\log\log(1/\epsilon)} \right)
\end{equation}
\end{proof}
Next in order to use this in our following lemmas, we propose the following ancilla-free method for constructing a block-encoding of the logarithm of a matrix.
\begin{lemma} \label{lem:BERotAngle}
Let $\theta \le 3/10$ and let $\epsilon$ be a positive number such that $\epsilon < 1- \arcsin(1/3)$ then given access to an oracle that yields $e^{-iZ\theta}$ there exists an algorithm that yields a unitary $V\in \mathbb{C}^{2\times 2}$ such that 
$
|\bra{0}V\ket{0} - \theta|\le \epsilon
$
using no ancillary qubits and a number of queries to $e^{-iZ\theta}$
and additional single qubit operations that scales as 
$$
\mathcal{O}\left(\log\left(\frac{1}{\epsilon} \right) \right)
$$
\end{lemma}
\begin{proof}
Our strategy is to provide an inexpensive way of block-encoding the rotation angle in the upper-left corner of our matrix.  We achieve this by recognizing that
\begin{equation}
e^{iX \pi/2}\begin{bmatrix}
    \cos(\theta) & -i \sin(\theta)\\
    -i\sin(\theta) & \cos(\theta)
\end{bmatrix} = \begin{bmatrix}
    \sin(\theta) & -i\cos(\theta)\\
    -i\cos(\theta) &\sin(\theta)
\end{bmatrix}
\end{equation}
Thus we can apply QSP to convert this rotation to having $\theta$ in the top left hand block by applying a polynomial that approximates $\arcsin(a)$ to the top left entry $a:=\sin(\theta)$ in this context.  The central problem here is that there cannot be a polynomial that precisely satisfies the QSP assumptions needed because all such functions must have range within $[-1,1]$.  Arcsine in contrast has range $[-\pi/2,\pi/2]$. 

To this end, let us restrict the range such that $|\theta| \le 1/3$.  In this case we have that
\begin{equation}
    |\arcsin(\theta)| \le \arcsin(3/10)\approx 0.305.
\end{equation}
Under these assumptions we can instead implement the product of a top-hat function and the arcsine function on the slightly larger interval $[-1/\pi,1/\pi]$
\begin{equation}
    P(x) = (H(x+1/\pi)-H(x-1/\pi)) \arcsin(x):= W(x) \arcsin(x)
\end{equation}
where $H(x)$ is the Heaviside function.  It is then easy to see that $P(x) \in [-1,1]$ for all $x\in [-1,1]$.

First let us consider the higher-order derivatives of $\arcsin(x)$.  From the Taylor series of $\arcsin$ we see that for any $|x|< L$ we can use the Gamma-function to express
\begin{align}
    |\partial_x^q \arcsin(x)| &=\left|\partial_x^q \sum_{k=0}^\infty \frac{x^{1+2k} \Gamma(k+1/2)}{k!(1+2k)\Gamma(1/2)}\right|\nonumber\\
    &\le \sum_{k\ge q/2 -1} \frac{|x^{1+2k-q}|(1+2k)! \Gamma(k+1/2)}{k!(1+2k)\Gamma(1/2) (1+2k-q)!}\nonumber\\
    &\le \sum_{k\ge q/2 -1} \frac{L^{1+2k-q}(1+2k)! \Gamma(k+1/2)}{k!(1+2k)\Gamma(1/2) (1+2k-q)!}\nonumber\\
    &\le q!\sum_{k\ge q/2 -1} \frac{L^{1+2k-q}}{(1+2k)}\binom{1+2k}{q}\nonumber\\
    &\le q!\sum_{k\ge q/2 -1} {L^{1+2k-q}}{}\binom{1+2k}{q}\nonumber\\
    &\le \sum_{k\ge q/2 -1} {L^{1+2k-q}}{}(1+2k)^q\nonumber\\
      &\le \sum_{u\ge 0} {L^{u}}{}(u+q)^q\nonumber\\
&\le q^q\sum_{u\ge 0} {L^{u}}e^u\nonumber\\
&\le \frac{q^q}{1-eL }
\end{align}

Using a Lagrange interpolatory polynomial of degree $Q-1$ (denoted $P_Q$) using a uniform grid on the interval $x\in [-L,L]$, we can approximate arcsine x we see that the remainder term in the series approximation, $P_Q(x)$, truncated at order $Q$ is from Taylor's remainder theorem
\begin{align}
|P_Q(x) - \arcsin(x)| &\le \frac{\max_x |\partial_x^{Q+1} \arcsin(x)||\prod_{j=0}^Q (x-(-L+2Lj/Q)) | }{(Q+1)!}\nonumber\\
&\le \frac{(Q+1)^{Q+1}((Q+1)!(2L/Q)^{Q+1})}{(Q+1)!(1-eL)}\nonumber\\
&\le \frac{2(2L)^{Q+1} (1+1/Q)^Q}{1-eL}\nonumber\\
&\le \frac{2e(2L)^{Q+1}}{1-eL}\label{eq:arcsinD}
\end{align}
Thus we can prove exponentially decreasing error with $Q$ for any $L< e^{-1}$ and our restriction that $L\le 1/\pi$ is taken because $1/\pi < e^{-1}$ (although tighter bounds could be realized) and $1/\pi > 3/10$. This latter bound is taken to provide our approximation to the Heaviside function some ``padding" that we will later need to guarantee that the resultant polynomial is bounded within $[-1,1]$.

We can then use~\eqref{eq:arcsinD} to show that if we wish $|P_Q(x) -\arcsin(x)| \le \epsilon$ on the interval $[-1/\pi,1/\pi]$ then it suffices to pick
\begin{equation}
    Q=  \left\lceil\frac{\log\left(\frac{2e}{(1-e/\pi)\epsilon)}\right)}{\log(e/2)} \right\rceil-1 \le \alpha \log_2(1/\epsilon), 
\end{equation}
where $\alpha$ is a universal constant. 
Thus a logarithmic degree polynomial suffices to closely approximate the function over this restricted interval.

Next we would like to consider what happens with the polynomial  $P_Q$ outside of the interval $[-1/\pi,1/\pi]$.   As $Q$ is arbitrary at this stage,  let us take $Q$ to be even.  For even $Q$ we have that the derivative of the degree $Q+1$ polynomial can be bounded by
\begin{align}
    |\partial_x P_Q(x)| &\le  \sum_{j=0}^{Q} \arcsin(-1/\pi + 2j/\pi Q) \left|\partial_x\prod_{k=0,k\ne j}^Q \frac{(x-x_k)}{(x_j-x_k)}\right| \nonumber\\
    &\le Q(Q+1) \arcsin(1/3) \max_{j,\ell} \prod_{k=0,k\ne j,\ell}^Q \frac{(x-x_k)}{(x_j-x_k)}\nonumber\\
    &\le Q(Q+1) \arcsin(1/\pi)  \frac{(2/\pi)^{Q-1}}{\min_j\prod_{k=0,k\ne j}^Q(x_j-x_k)}\nonumber\\
    &\le Q(Q+1) \arcsin(1/\pi)  \frac{(2/\pi)^{Q-1}}{(\prod_{k=1}^{Q/2}(2k/\pi Q))^2}\nonumber\\
    & =\frac{\pi Q(Q+1) \arcsin(1/\pi)}{2}  \frac{(2/\pi)^{Q}Q^Q}{(Q/2)!^2}\nonumber\\
    & \le\frac{\pi Q(Q+1) \arcsin(1/\pi)}{2}  \frac{(1/\pi)^{Q}Q^Q}{\left(\frac{Q/2}{e} \right)^Q}\nonumber\\
    &=\frac{\pi Q(Q+1)(2e/\pi)^{Q} \arcsin(1/\pi)}{2} 
\end{align}
Thus the maximum error that is observable a distance $y$ away from the last interpolation point is
\begin{equation}
\mathcal{O}\left({yQ^2(2e/\pi)^Q} \right)\subseteq \mathcal{O}\left(2^Q y \right)\subseteq \mathcal{O}\left(\frac{ y}{\epsilon^{\alpha}} \right)
\end{equation}
Thus we need to ensure that the function $W$ achieves a value that is at most $\epsilon^{\alpha}$ outside the interval of interest.

Using existing results from polynomial approximation contained in~\cite{Low_2017}, there exists a polynomial approximation to $W$, $P_W$, that maps $[-1,1]\mapsto [-1,1]$  attains a value of $\epsilon'$ for all $x\ge 1/\pi+\delta'$ such that
\begin{equation}
    {\rm deg}(P_W) \in \mathcal{O}\left(\frac{\log(1/\epsilon')}{\delta'} \right)
\end{equation}
We then see that if we pick
\begin{eqnarray}
    \delta' \le \frac{\pi - 3}{3\pi}, \quad \epsilon' \in \mathcal{O}\left(\epsilon^\alpha \right). 
\end{eqnarray}
From this we can see immediately that
\begin{equation}
    {\rm deg}(P_W P_Q)= {\rm deg}(P_W) + {\rm deg}(P_Q)\in \mathcal{O}\left(\log\left(\frac{1}{\epsilon}\right) \right) 
\end{equation}
then it is easy to see from the triangle inequality that
\begin{equation}
    |P_Q(x) P_W(x) - \arcsin(x) (H(x+1/\pi) - H(x-1/\pi))| \le \begin{cases} |x| \le 1/\pi, &\mathcal{O}(\epsilon + \epsilon') \subseteq \mathcal{O}(\epsilon)\\
    |x| \in (1/\pi, 1/3], & \mathcal{O}(1) \\
    |x| \in (1/3, 1], & \mathcal{O}(\epsilon'/\epsilon^\alpha) \subseteq \mathcal{O}(\epsilon)\end{cases}.
\end{equation}
We then see that if $\epsilon$ is sufficiently small then the conditions needed for GQSP implementation of the product $P_Q P_W$ is possible, if we can show that the value of the approximation $|P_Q(x) P_W(x)|\le 1$ for all $x$ such that $|x| \in (1/\pi,1/3]$.  To see this note that
\begin{equation}
    |P_Q(x)|\le \arcsin(1/3) + \epsilon~\forall~ x\in [-1/\pi,-1/3)\cup (1/3,1/\pi ], \quad |P_W(x)| \le 1~\forall~x\in [-1/\pi,-1/3)\cup (1/3,1/\pi]
\end{equation}
Therefore provided $\epsilon \le 1-\arcsin(1/3)$ we have that 
\begin{equation}
|P_Q(x) P_W(x)| \le 1~\forall~x\in [-1/3,1/\pi)\cup (1/\pi,1/3]
\end{equation}
\end{proof}

We can now combine Lemma \ref{lem:JacobiAngerSq} and \ref{lem:BERotAngle} to square small input angles where $| \theta | \leq 1/3$ as provided by our oracle. 

\begin{corollary}\label{cor:PUPSmall}
Assume that we are given an oracle, $W_Z(\theta)$, that implements $e^{iZ\theta/2}$ for $|\theta|\le 3/5$. We can then construct a unitary $V\in \mathbb{C}^{2\times 2}$ such that $\|V-\exp(i(\frac{\theta^2}{2}) Z)\|\le \epsilon$ where using
a number of applications of an input oracle rotation $W_Z(\theta)$  and single qubit unitaries 
$$
\mathcal{O}\left(\frac{\log^2(1/\epsilon)}{\log\log(1/\epsilon)} \right)
$$
\end{corollary}
\begin{proof}
Lemma~\ref{lem:JacobiAngerSq} shows that for error $\epsilon'$ there exists a polynomial of degree $M^* \in \mathcal{O}(\log(1/\epsilon')/\log\log(1/\epsilon'))$ such that  an application of Theorem~\ref{thm:GQSP} will allow us to implement a transformation on the input oracle $R_x(\arccos(a))$ denoted $\Lambda_{2} : \mathbb{C}^{2\times 2} \mapsto \mathbb{C}^{2\times 2}$ such that
\begin{equation}
    \Lambda_2 \left(\begin{bmatrix}
        a & -i\sqrt{1-a^2}\\
        -i\sqrt{1-a^2} & a 
    \end{bmatrix}\right)=
    e^{ia^2 Z} +\epsilon' E
\end{equation}
for an operator $E$ with norm at most $1$ while using a number of queries to $W_z(\theta)$ and gate operations in $\mathcal{O}(M^*)=\mathcal{O}(\log(1/\epsilon')/\log\log(1/\epsilon'))$. 

Similarly Lemma~\ref{lem:BERotAngle} shows that there exists a polynomial of degree at most $M \in \mathcal{O}(\log(1/\epsilon''))$ such that the result of Theorem~\ref{thm:GQSP} probides a construction that yields a transformation on an inital oracle $W_z(\theta)$ denoted $\Lambda_{\arcsin}:\mathbb{C}^{2\times 2} \mapsto \mathbb{C}^{2\times 2}$ such that for $\theta\le 3/5$
\begin{equation}
    \Lambda_{\arcsin}(e^{-iZ\theta/2}) = \begin{bmatrix}
        \frac{\theta}{2} & - i \sqrt{1- \frac{\theta^2}{4}}\\
        - i \sqrt{1- \frac{\theta^2}{4}} & \frac{\theta}{2}
    \end{bmatrix} + \epsilon'' F
\end{equation}
where $\|F\|\le 1$ using a number of single qubit operations and queries to $W_z$ that scale as $\mathcal{O}(\log(1/\epsilon''))$.  

We then see that that for every $\epsilon>0$ there exists $\epsilon'$ and $\epsilon''$ such that
\begin{equation}
    \Lambda_2( \Lambda_{\arcsin}( e^{-i Z\theta/2})) = e^{-iZ (\theta/2)^2} + \epsilon G,
\end{equation}
where $\|G\|\le 1$.
Thus from Box 4.1 from~\cite{nielsen2001quantum} the error of these two composed sequences of rotations is at most additive.   We then have that in order to ensure that the overall error in the algorithm is $\epsilon$ it suffices to take $\epsilon' = \epsilon$ and $\epsilon'' \in \Theta(\epsilon/ M^*)$ which implies that the overall cost of composing these two channels is
\begin{equation}
    \mathcal{O}\left(\frac{\log(\log(1/\epsilon)/\epsilon) \log(1/\epsilon)}{\log\log(1/\epsilon)} \right) = \mathcal{O}\left(\frac{\log^2(1/\epsilon)}{\log\log(1/\epsilon)} \right) 
\end{equation}
\end{proof}

\begin{lemma}[Generalized Multiplication of Phases]\label{lem:PMP}
Assume that you are provided a set of signal operators $\{W_Z(\theta_i) : i =1,\ldots,N\}$ and a set of non-negative integers $\{\ell_i: i=1\ldots N\}$ then we can construct a block encoding of $\exp(i\prod_{i=1}^{2^\mu}(\frac{\theta_i}{2})^{\ell_i}Z)$ using a number of queries to $W_Z(\theta_i)=e^{i \theta_i Z/2}$ that obey the following scalings.
\begin{enumerate}
    \item If each $|\theta_i| \le 3/5$ and each $\ell_i=1$ then there exists a protocol that uses no ancillary qubits and yields an operator $V$ such that (up to global phases) $\|V -\exp(i\prod_{i=1}^{2^\mu}(\frac{\theta_i}{2})^{\ell_i}Z)\|\le \epsilon$ using a number of applications of the $W_z$ oracles and single qubit operations that is in
    $$\mathcal{O}\left(\left(\frac{4\kappa\log^2(1/\epsilon')}{\log\log(1/\epsilon')}\right)^{\mu}\right)$$
    \item Otherwise there exists an algorithm that yields a  $(1-O(\epsilon),O(\max(\ell_i)),\epsilon)$ block encoding of $\exp(i\prod_{i=1}^{2^\mu}(\frac{\theta_i}{2})^{\ell_i}Z)$ and a universal constant $\kappa\ge 1$ such that the total number of single qubit operations and controlled queries to $W_z(\theta_i)$ needed scales as
    $$
\mathcal{O}\left(\left(\frac{4\kappa\log^2(1/\epsilon')}{\log\log(1/\epsilon')}\right)^{\mu} \max_{l'\in \{\ell_i\}}\frac{(\frac{\pi}{2})^{l'}\log^2(l'/\epsilon')}{\log\log(l'/\epsilon')} \right).
    $$
\end{enumerate}
\end{lemma}
\begin{proof}
The first stage in the proof involves showing that we can implement multiplication in phase.  We can see how to do this using the following observation
\begin{equation}
    \left(\frac{\theta_{1}}{2} + \frac{\theta_{2}}{2}\right)^2 = \left(\frac{\theta_{1}}{2}\right)^{2} + \left(\frac{\theta_{2}}{2}\right)^{2} + \frac{1}{2} \theta_{1}\theta_{2} 
\end{equation}
Thus we can express the product between two results
\begin{equation}
    \frac{1}{2} \theta_1 \theta_2 = \left(\frac{\theta_{1}}{2} + \frac{\theta_{2}}{2}\right)^2-\left(\frac{\theta_{1}}{2}\right)^{2}-\left(\frac{\theta_{2}}{2}\right)^{2}.
\end{equation}
Next note that
\begin{equation}
    W_Z(\theta_1) W_Z(\theta_2) = W_Z(\theta_1+\theta_2).
\end{equation}
Thus we can see that we can implement this using Lemma~\ref{lem:square} to implement a $(1-3\epsilon',6,\epsilon')$ block encoding of
\begin{equation}
    e^{i \theta_1\theta_2 Z/2} = e^{i \left(\frac{\theta_{1}}{2} + \frac{\theta_{2}}{2}\right)^2 Z}Xe^{i (\theta_1/2)^2 Z}e^{i (\theta_2/2)^2 Z}X = W_Z(\theta_1\theta_2).
\end{equation}
using a number of  signal operations or rotations, $N_W$, that scale as
\begin{equation}\label{eq:NW}
    N_W = \mathcal{O}\left( \frac{\log^2(1/\epsilon')}{\log\log(1/\epsilon')}\right). 
\end{equation}
This shows that we can implement the product between two rotations using the squaring algorithm.

Let us assume for the moment that we wish to implement for some constant $\mu\in \mathbb{Z}_+$
\begin{equation}
W_Z(\prod_{i=1}^{2^\mu} \theta_i).
\end{equation}
Now let us define the cost of implementing a product of $2^{\nu}$ phases to be $C_\nu$.  Let us then multiply each of the phases using a divide and conquer strategy using splits of size $1/2$.  Specifically let ${\rm SQUARE}$ be the phase squaring method of Lemma~\ref{lem:square} and assume that the number of queries needed to implement the method is at most $\kappa \log^2(1/\epsilon')/\log\log(1/\epsilon')$
\begin{eqnarray}
    {\rm SQUARE}\left( W_Z(\prod_{i=1}^{2^{\nu-1}} \theta_i)W_Z(\prod_{i=2^{\nu-1}+1}^{2^{\nu}} \theta_i)\right){\rm SQUARE}^{\dagger}\left( W_Z(\prod_{i=1}^{2^{\nu-1}} \theta_i))\right){\rm SQUARE}^{\dagger}\left( W_Z(\prod_{i=2^{\nu-1}+1}^{2^{\nu}} \theta_i)\right)
\end{eqnarray}
Thus the cost of implementing this is
\begin{equation}
    C_{\nu} = \frac{4\kappa\log^2(1/\epsilon')C_{\nu-1}}{\log\log(1/\epsilon')}
\end{equation}
If we define the maximum cost of performing one of the fundamental $W(\theta_i)$ to be $C_0$ then the solution to the recursion relation is
\begin{equation}\label{eq:Cmu}
    C_{\mu} \le \left(\frac{4\kappa\log^2(1/\epsilon')}{\log\log(1/\epsilon')}\right)^{\mu} C_0.
\end{equation}

Our ultimate target is to implement the generalized product 
$\exp(i\prod_i{\theta_i^{\ell_i}}Z/2)
$.  We can implement this by the product between each of the individual phases using the result of Corollary~\ref{cor:PUP} to find that 
\begin{equation}\label{eq:C0bd}
    C_0 \le \mathcal{O}\left(\max_{l'\in \{\ell_i\}}\frac{(\frac{\pi}{2})^{l'}\log^2(l'/\epsilon')}{\log\log(l'/\epsilon')} \right),
\end{equation}
We then see that
\begin{equation}
    C_\mu \le \mathcal{O}\left(\left(\frac{4\kappa\log(1/\epsilon')}{\log\log(1/\epsilon')}\right)^{\mu} \max_{l'\in \{\ell_i\}}\frac{(\frac{\pi}{2})^{l'}\log(l'/\epsilon')}{\log\log(l'/\epsilon')} \right)
\end{equation}

Now let us consider the case where the input rotation angles are promised to be smaller than $2/3$.  If this is the case then we can use the result of Corollary~\ref{cor:PUPSmall} to build a squaring circuit.  In this case we also assume that the powers $\ell_i$ are 1.  Next assume that the cost of performing a squaring operation is at most $\kappa \log^2(1/\epsilon') / \log\log(1/\epsilon')$.  Let $C_\nu$ be the cost of multiplying $2^{\nu}$ phases.  We then have that for any $\nu \ge 0$
\begin{equation}
    C_{\nu+1} = \left(\frac{4\kappa \log^2(1/\epsilon')}{\log\log(1/\epsilon')}\right)C_{\nu}.
\end{equation}
Let us assume that the error in layer $\nu$ is $\epsilon_\nu$.  We then have from the sub-additivity of error given in Box 4.1 of~\cite{nielsen2001quantum} that the error at layer $\nu+1$ in the formula is at most $\epsilon_{\nu}$ then we can choose $\epsilon' = \epsilon_\nu$
\begin{equation}
    \epsilon_{\nu+1}=\epsilon' + \left(\frac{4\kappa \log^2(1/\epsilon')}{\log\log(1/\epsilon')}\right) \epsilon_\nu \le  \left(\frac{5\kappa \log^2(1/\epsilon_\nu)}{\log\log(1/\epsilon_\nu)}\right) \epsilon_{\nu}
\end{equation}
The solution since the first non-trivial error occurs at the first layer of the multiplication we have that the error obeys
\begin{equation}
    \epsilon_\nu \le \prod_{j=1}^{\nu-1}\left(\frac{5\kappa \log^2(1/\epsilon_j)}{\log\log(1/\epsilon_j)}\right) \epsilon_1
\end{equation}
as $\epsilon_j$ is a monotonically increasing sequence we have that
\begin{equation}
    \epsilon_\nu \le \left(\frac{5\kappa \log^2(1/\epsilon_1)}{\log\log(1/\epsilon_1)}\right)^{\nu-1} \epsilon_1 \le (5\kappa \log^2(1/\epsilon_1))^{\nu-1} \epsilon_1 = \epsilon
\end{equation}
The solution to this is written in terms of the Lambert-W function as
\begin{equation}
    \epsilon_1 = e^{2(\nu-1) W(-\epsilon^{1/2(\nu-1)}/(5\kappa(2\nu-2)))} = \mathcal{O}\left(\epsilon/\log^{2(\nu-1)}(1/\epsilon)\right)
\end{equation}

We can then substitute the above result into the cost to find an asymptotic bound on the complexity via the following steps.  First, 
\begin{equation}
    C_{\nu} = \left(\frac{4\kappa \log^2(1/\epsilon_{\min})}{\log\log(1/\epsilon_{\min})}\right)^\nu C_{0}.
\end{equation}
Next by definition the product of $1$ angle requires $1$ query.  Thus the total cost is 
\begin{align}
    C_{\nu} &= \left(\frac{4\kappa \log^2(1/\epsilon_{\min})}{\log\log(1/\epsilon_{\min})}\right)^\nu \nonumber\\
    &= \mathcal{O}\left( \left(\frac{4\kappa \log^2(\log^{2(\nu-1)}(1/\epsilon)/\epsilon)}{\log\log(\log^{2(\nu-1)}(1/\epsilon)/\epsilon)}\right)^\nu\right)\nonumber\\
    &= \mathcal{O}\left( \left(\frac{4\kappa \log^2(1/\epsilon)}{\log\log(1/\epsilon)}\right)^\nu\right)
\end{align}
\end{proof}

Finally, with these results it is easy to see how an arbitrary polynomial can be generated.  We state this result below

\begin{corollary}[Implementation of Phase Polynomials]
Assume that you are provided a set of signal operators $\{W_Z(\theta_i) : i =1,\ldots,N\}$ and a set of non-negative integers $\{\ell_{i,j}: i=1\ldots N, j=1\ldots M\}$ and coefficients $\{\alpha_{j}: j=1,\ldots,M\}$ then we can construct a block encoding of $\exp(i\sum_{j=1}^M a_j\prod_{i=1}^{2^{\mu_j}}(\frac{\theta_{i,j}}{2})^{\ell_{i,j}}Z)$ using a number of queries to $W_Z(\theta_i)=e^{i \theta_i Z/2}$ that obey the following scalings.
\begin{enumerate}
    \item If each $|\theta_i| \le 3/5$ and each $\ell_i=1$ then there exists a protocol that uses no ancillary qubits and yields an operator $V$ such that (up to global phases) $\|V -\exp(i\sum_{j=1}^M\prod_{i=1}^{2^{\mu_j}}(\frac{\theta_i}{2})^{\ell_i}Z)\|\le \epsilon$ using a number of applications of the $W_z$ oracles and single qubit operations that is in
    $$\widetilde{\mathcal{O}}\left(M\left(\frac{4\kappa\log^2(1/\epsilon')}{\log\log(1/\epsilon')}\right)^{\max_j\mu_j}\right)$$
    \item Otherwise there exists an algorithm that yields a  $(1-O(\epsilon),O(\max(\ell_i)),\epsilon)$ block encoding of $\exp(i\sum_{j=1}^M\prod_{i=1}^{2^{\mu_j}}(\frac{\theta_i}{2})^{\ell_i}Z)$ and a universal constant $\kappa\ge 1$ such that the total number of single qubit operations and controlled queries to $W_z(\theta_i)$ needed scales as
    $$
\widetilde{\mathcal{O}}\left(M\left(\frac{4\kappa\log^2(1/\epsilon')}{\log\log(1/\epsilon')}\right)^{\max_j\mu_j} \max_{l'\in \{\ell_{i,j}\}}\frac{(\frac{\pi}{2})^{l'}\log^2(l'/\epsilon')}{\log\log(l'/\epsilon')} \right).
    $$
\end{enumerate}
\end{corollary}
\begin{proof}
The proof of this result is a trivial consequence of Lemma~\ref{lem:PMP}.  First note that while we have weights $a_j$ we can take $a_j=1$ without loss of generality because we can always introduce a fictitious $\theta_{i,j} = a_j$.  This at most increases the cost by a constant factor.  In turn, the error tolerance for each gate needed per term increases by a constant factor, which is irrelevant to the $\tilde{\mathcal{O}}(\cdot)$ scaling.

Next note that we can add such polynomials by simply repeating each monomial in series.  There are $M$ such monomials and so the cost is $M$ times the maximum of the previous costs.  However, the error per exponential needs to be reduced by a factor of $M$ to ensure that the $\epsilon$ target of the block-encoding is reached.  As the $\epsilon$-scaling of both versions of the cost statements in Lemma~\ref{lem:PMP} is logarithmic, this leads to multiplicative logarithmic factors.  Such factors do not appear in the $\widetilde{\mathcal{O}}(\cdot)$ notation above and so are dropped from the final cost estimates.
\end{proof}

\subsection{Applications of M-QSP to Coulomb Potential Calculations}
In this section we show how our M-QSP scheme can be applied to a simple but important example, the two-dimensional Coulomb potential $V(x,y) = 1/\sqrt{(x_1-x_2)^2 + (y_1 - y_2)^2}$. Suppose we are given oracle access to four unitaries (and their inverses)
\begin{equation}
\lbrace e^{ix_{1} Z}, e^{iy_{1} Z}, e^{ix_{2} Z}, e^{iy_{2}Z} \rbrace
\end{equation}
where the $(x_{1}, y_{1})$ and $(x_{2}, y_{2})$ correspond to the positions of point charge $1$ and point charge $2$ respectively appropriately scaled by a scaling factor $\alpha$ to $(x_{1}, y_{1})$ and $(x_{2}, y_{2})$ such that $|x_{i}| \leq \frac{1}{4}$ and $|y_{i}| \leq \frac{1}{4}$. This is to ensure that the spectral norm condition of our phase squaring algorithm is met.
First, we note that $e^{i (x_{1} - x_{2}) Z} = e^{ix_{1} Z} e^{-ix_{2} Z} $ and $e^{i (y_{1} - y_{2}) Z} = e^{iy_{1} Z} e^{-iy_{2} Z} $, both products are achievable through direct multiplication of the respective signal unitaries. Then, using our phase squaring algorithm on $e^{i (x_{1} - x_{2}) Z}$, we can obtain a block-encoding of $e^{i (x_{1} - x_{2})^{2} Z}$ and similarly for $e^{i (y_{1} - y_{2})^{2} Z}$. By block-encoded multiplication of these two resulting unitaries, we can thus obtain $e^{i \left[(x_{1} - x_{2})^{2} + (y_{1} - y_{2})^{2} \right] Z}$. The following theorem  (stated as Corollary 67 in~\cite{Gily_n_2019}) will be useful for the next step.
\begin{theorem}[Polynomial approximations of negative power functions]
\label{thm:polyapproxNPF}
    Let $\epsilon, \delta \in (0, \frac{1}{2} ]$ and let $f(x) := \frac{\delta^{c}}{2}x^{-c}$, then there exist even/odd polynomials $P$,$P' \in \mathbb{R}[x]$, such that $\norm{ P - f }_{[\delta,1]} \leq \epsilon$, $\norm{ P }_{[-1,1]} \leq 1$ and similarly $\norm{ P' - f }_{[\delta,1]} \leq \epsilon$, $\norm{ P' }_{[-1,1]} \leq 1$, moreover the degree of the polynomials are $\mathcal{O} \left( \frac{\max[1,c]}{\delta} \log \left( \frac{1}{\epsilon}\right) \right)$.
\end{theorem}
Note that the choice of a fixed $\delta \in (0 ,\frac{1}{2}]$ necessitates the promise that our oracles (corresponding to the four signal unitaries) collectively guarantee that $(x_{1} - x_{2})^{2} + (y_{1} - y_{2})^{2} \geq \delta$ for our method to work. Thus, by taking $c = \frac{1}{2}$ and with an appropriately small choice of a fixed $\delta$ that meets the range requirements of our polynomial argument, we can therefore implement the function $f(x) := \frac{\sqrt{\delta}}{2} x^{-\frac{1}{2}}$ via an appropriate approximating polynomial in the following steps. First, we perform a logarithmic block-encoding of our previously obtained $e^{i \left[(x_{1} - x_{2})^{2} + (y_{1} - y_{2})^{2} \right] Z}$, thereby recovering $[(x_{1} - x_{2})^{2} + (y_{1} - y_{2})^{2}] Z$. Now,  Theorem \ref{thm:polyapproxNPF} can be used to obtain a polynomial block-encoding of $[(x_{1} - x_{2})^{2} + (y_{1} - y_{2})^{2}]^{-\frac{1}{2}}Z$ with degree $\mathcal{O} \left( \frac{1}{\delta} \log \left( \frac{1}{\epsilon}\right) \right)$. If we then perform a measurement in the computational basis, either measurement outcome will give the correct magnitude for the corresponding Coulomb potential $1/\sqrt{(x_{1} - x_{2})^{2} + (y_{1} - y_{2})^{2}}$ between the two point charges as required. 


\section{Application: Iterated QSP and Bosonic Simulations}\label{sec:boson}
In our following application, we will use  the $(W_{X}, S_{Z}, \langle + | \cdot | + \rangle)$-QSP convention.  Under this scheme, we would like to obliviously transform our given signal unitary $W_{X} \equiv e^{-i\theta \tau X}
$ into $e^{-i\sqrt{\theta}\tau X}$. 
Similar to the previous strategies used in this paper, we will first determine an appropriate approximating polynomial for our desired transformation
and then apply appropriate QSP phases corresponding to this approximating polynomial to our signal rotation operators interleaved alternately with a (block-encoded) signal unitary. Before going into technical details, we would like to provide a brief description of the the physical problem here. 


Simulating bosonic systems on quantum computers is a surprisingly difficult problem in part because of the fact that bosons do not obey the Pauli exclusion principle.  The Pauli exclusion principle guarantees that at most one particle can be in a given orbital.  This implies that occupations in fermionic systems, can be tracked with a single qubit.  Bosonic systems on the other hand can have in principle an unbounded Hilbert-space dimension.  In particular, we have that the boson number in a mode is given by
\begin{equation}
    n=a^\dagger a = \sum_j j \ketbra{j}{j},
\end{equation}
where in occupation basis we have that
\begin{equation}
    a^\dagger \ket{j} = \sqrt{j+1}\ket{j+1}
\end{equation}
These operators are non-Hermitian and typically arise in simulations with their Hermitian conjugates in a form similar to
\be 
H = a+ a^\dagger= \mqty(0 & 1 & 0 & 0 & . & . & 0 \\ 1 & 0 & \sqrt{2} & 0 & . & . & 0\\ 0 & \sqrt{2} & 0 & \sqrt{3} & . & . & 0\\ 0 & 0 & \sqrt{3} & 0 & . & . & 0 \\ . & . & . & . & . & . & . \\ . & . & . & . & . & 0 & \sqrt{n} \\ . & . & . & . & . & \sqrt{n} & 0 )
\label{eq:boson_hamiltonian}
\ee 
The primary challenge that arises when simulating this Hamiltonian stems from the presence of the square-roots.  Conventional approaches to computing square roots on quantum computers rely on using reversible circuits to compute a binary representation of the bitstring~\cite{HanerSvore}.  While these algorithms are efficient, the number of qubits needed to compute the square root can be prohibitively large using methods such as Newton iteration.  Our aim is to provide a new way of evaluating the square roots in phase that uses only a constant number of qubits and is not limited in the precision of a fixed point representation.

Our approach for performing the simulation involves performing a decomposition of the element-wise square of such matrices into a sum of one-sparse Hamiltonians.  We then use iterated QSP to compute square roots of the matrix elements of the sparse matrices.

\subsection{Simulation of Banded Matrices}
In order to simulate Eq.~\eqref{eq:boson_hamiltonian}, we take a very different approach. We start from the Hamiltonian simulation of banded matrices. The advantage of banded matrices is that they are sparse and can be easily decomposed into 1-sparse matrices that can be simulated using a single query and one ancilla. To see this, let us consider a Hamiltonian $H$ of the form,
\be 
H = \mqty(0 & 1 & 0 & 0 & . & . & 0 \\ 1 & 0 & 2 & 0 & . & . & 0\\ 0 & 2 & 0 & 3 & . & . & 0\\ 0 & 0 & 3 & 0 & . & . & 0 \\ . & . & . & . & . & . & . \\ . & . & . & . & . & 0 & n \\ . & . & . & . & . & n & 0 )
\ee 
We want to do Hamiltonian simulation $e^{-iHt}$ using sparse decomposition of the Hamiltonian matrix.
We note that $H$ can be decomposed into two segments each containing $1$-sparse matrices, $H_A$ and $H_B$,
\bea 
H_A = \mqty(0 & 1 & 0 & 0 & . & . & 0 \\ 1 & 0 & 0 & 0 & . & . & 0\\ 0 & 0 & 0 & 3 & . & . & 0\\ 0 & 0 & 3 & 0 & . & . & 0 \\ . & . & . & . & . & . & . \\ . & . & . & . & . & 0 & n \\ . & . & . & . & . & n & 0 ) ;\;
H_B = \mqty(0 & 0 & 0 & 0 & . & . & 0 \\ 0 & 0 & 2 & 0 & . & . & 0\\ 0 & 2 & 0 & 0 & . & . & 0\\ 0 & 0 & 0 & 0 & . & . & 0 \\ . & . & . & . & . & . & . \\ . & . & . & . & . & 0 & 0 \\ . & . & . & . & . & 0 & 0 )
\eea 
Once we have $1$-sparsified $H$, our next goal is to transform the block-diagonal Hamiltonian into single qubit forms. In order to do this, we make use of the fact that each block resembles $\sigma_x$, thus performing $e^{-it H_{A}}$ is,
\be
e^{-it H_{A}} = e^{-it X}\oplus e^{-i3tX}..\oplus e^{-intX}
\label{eq: ha_evolution}
\ee 
With the simple form of eq.~\eqref{eq: ha_evolution}, the key idea is to transfer the information of the original Hilbert space into single qubit, then apply each of the rotations on a single qubit and then transform it back to the original Hilbert space. 
In order to do this, let us first define an oracle $O_f$, such that, 
\be 
O_f\ket{x}\ket{y} = \ket{x}\ket{y \oplus f(x)},
\label{eq:oracle}
\ee 
where $f(x)$ is the index of the non-zero element in row $x$ (note, $H$ is $1$-sparse). Since $H$ is Hermitian, $f(f(x)) = x$. If we consider the case of $H_A$, we may write $f$ as,
\bea
    f(x) &=& x-1\;:\; \text{if}\; x \; \text{is even} \notag \\
    &=& x+1 \;:\; \text{if}\; x \; \text{is odd} 
\eea 
Taking the advantage of the simple form of the Hamiltonian structure, we can construct an oracle without any ancilla qubit using the following circuit, 
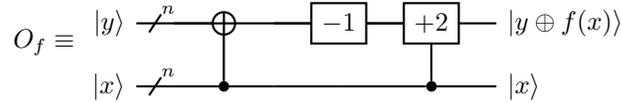
\begin{center}$O_f \equiv$ 
\begin{quantikz}
\lstick{$\ket{y}$} &\qwbundle{n} & \targ{0} \qw & \qw & \gate{-1}\qw & \gate{+2}\qw &\rstick{$\ket{y \oplus f(x)}$}\qw \\
\lstick{$\ket{x}$} &\qwbundle{n}  & \ctrl{-1} & \qw & \qw &\ctrl{-1} & \rstick{$\ket{x}$}\qw 
\end{quantikz}
\end{center}
Now, let us consider a quantum state $\ket{\psi}$ in the 2-D block-diagonal subspace, 
\be
\ket{\psi} = \alpha \ket{x} + \beta \ket{y}.
\ee
By using the principle of linearity, $\ket{\psi}$ has the same dimension as the number of rows. 
Adding an ancilla qubit $\ket{a}$ initialized to $zero$ and applying the oracle $O_f$ on $\ket{\psi}$ gives,
\bea
O_f\ket{\psi}\ket{0}\ket{0}_{a} &=& O_f\left(\alpha \ket{x}\ket{0} + \beta \ket{y}\ket{0}\right)\ket{0}_{a} \notag \\
&=& \left(\alpha \ket{x}\ket{f(x)} + \beta \ket{y}\ket{f(y)}\right)\ket{0}_{a} \notag \\
&=& \left(\alpha \ket{x}\ket{f(x)} + \beta \ket{f(x)}\ket{f(f(x))}\right)\ket{0}_{a} \notag \\
&=& \left(\alpha \ket{x}\ket{f(x)} + \beta \ket{f(x)}\ket{x}\right)\ket{0}_{a} 
\label{eq:step_1}
\eea 
In the last step, we have choses $y=f(x)$
To this end, we define a comparator gate \texttt{CMP},
\be
\texttt{CMP}\ket{x}\ket{y}\ket{0} = \ket{x}\ket{y}\ket{y<x}
\ee 
Next, if we choose $y = f(x)$,  and apply \texttt{CMP} to the R.H.S, eq.\eqref{eq:step_1}, (assuming $f(x) > x$) we get,
\bea 
\texttt{CMP}\left[ \left(\alpha \ket{x}\ket{f(x)} + \beta \ket{f(x)}\ket{x}\right]\ket{0}_{a} \right]\notag \\
= \alpha \ket{x}\ket{f(x)}\ket{0}_{a} + \beta \ket{f(x)}\ket{x}\ket{1}_{a}
\label{eq:step_2}
\eea 
Next, we apply a controlled swap (\texttt{CSWP}) that interchanges between qubit registers $\ket{x}$ and $\ket{y}$ based on the ancilla $\ket{a}$ on eq,~\eqref{eq:step_2}
\bea
\alpha \ket{f(x)}\ket{x}\ket{0}_{a} + \beta \ket{f(x)}\ket{x}\ket{1}_{a} \notag\\
= \ket{x}\ket{f(x)} \left(  \alpha \ket{0}_{a} + \beta \ket{1}_{a} \right)
\label{eq:step_3}
\eea 
Thus we have transferred the coefficients ($\alpha, \beta$) of the state $\alpha \ket{x} + \beta \ket{f(x)}$ to a single ancilla qubit $\ket{a}$. 
As a next step, we need to apply the exponential of the relevant $2
\times 2$ block to the single ancilla. More precisely, we need to carry out the operation ($e^{-iH_{A}\tau} = e^{-ix\tau X}$) to the single ancilla. We can achieve this by using a set of controlled $R_X$ rotation as shown in Fig.~\ref{fig:U1_circuit}). We call this operation as $U_1$,
\begin{figure}
    \centering
    \begin{quantikz}
\lstick{$\ket{0}_a$} &  \gate{R_X(2\tau)}  \qw & \gate{R_X(2 \cdot 2\tau)} & \qw & \ldots & \gate{R_X(2 \cdot2^j\tau)} \qw \\
\lstick[wires=4]{$\ket{x}$}
& \ctrl{-1} & \qw & \qw & \ldots & \qw & \qw\\
& \qw &\ctrl{-2} & \qw  & \ldots &  \qw  & \qw\\
& \ldots \\
& \qw &\qw  & \qw &  \qw &\ctrl{-4} & \qw \\
\end{quantikz}
    \caption{$U_1$ implementation}
    \label{fig:U1_circuit}
\end{figure}
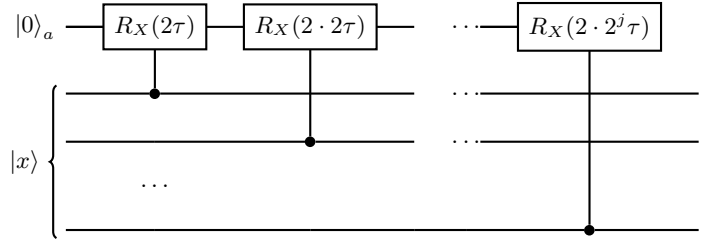
applying which on Eq.~\eqref{eq:step_3} will cause a conditional rotation of the ancilla state ($\alpha\ket{0} + \beta \ket{1}$),
\bea
U_{1}  \ket{x}\ket{f(x)} \left(  \alpha \ket{0}_{a} + \beta \ket{1}_{a} \right) &\rightarrow &
\ket{x}\ket{f(x)}e ^{-ix \tau X} \left(  \alpha \ket{0}_{a} + \beta \ket{1}_{a} \right) \\
&\equiv &  \ket{x}\ket{f(x)}R_{X}(2x\tau) \left(  \alpha \ket{0}_{a} + \beta \ket{1}_{a} \right)
\label{eq:U1}
\eea 
where,
\begin{equation}
    R_{X}(\tau) = \mqty[\cos (\frac{\tau}{2}) &-i \sin(\frac{\tau}{2}) \\ -i \sin(\frac{\tau}{2}) & \cos (\frac{\tau}{2})]
\end{equation}



\begin{figure}[t!]
\begin{center}
\begin{quantikz}
\lstick{$\ket{0}_a$} &\qw  & \qw &\qw   &\gate[3, nwires={3}, alternate]{\texttt{CMP}} &\ctrl{1}  &\gate[2]{U_1} &\ctrl{1} &\gate[3, nwires={3}, alternate]{\texttt{CMP}} & \qw \\
\lstick{$\ket{x}$}&\qwbundle{n} & \qw  & \gate[2]{O_f} \qw &\qw  & \gate[2]{\texttt{SWP}}\qw &  \qw & \gate[2]{\texttt{SWP}}\qw &\qw & \gate[2]{O_f}& \qw \\
\lstick{$\ket{y}$}&\qwbundle{n} &  \qw & \qw &  \qw  &\qw  & \qw & \qw &\qw \qw  &\qw &\qw\\
\end{quantikz}
\end{center}
\caption{Circuit for implementing an $X$ rotation on an irreducible two dimensional space using the unitary $U_1$.}
\end{figure}
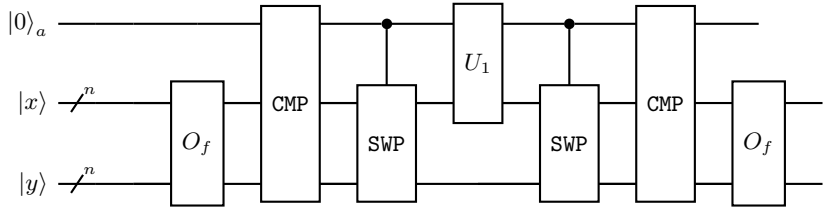

If we set an upper bound for $x<n_{cut}$, we will need $n_q \sim \log{n_{cut}}$ qubits, such that $x_{cut} = x_{0} + 2x_{1} + ... + 2^{n_{q}-1}x_{n_{q}}$, where $x_{j}\in \{0,1\}$.  Thus, number of controlled $R_{X}$ needed for a circuit without the square root will be $n_q = \log{n_{cut}}$.
Finally, by undoing the initial unitary operations we will obtain $ e^{-ix X\tau} \left(  \alpha \ket{x}  + \beta\ket{f(x)}  \right)\ket{0}\ket{0}_{a}$. Note that $e^{-ix \tau X}$ is the same as $e^{-i\tau H_{A}}$. By the superposition principle,  we finally obtain $e^{-i\tau H_{A}} \ket{\psi}\ket{0}\ket{0}_{a}$. Similarly, we can also simulate $e^{-i \tau H_{B}} \ket{\psi}$ and combine the two segments using a Product Formula(PF) approach \cite{berry2007efficient} to obtain $e^{-i \tau H_{B}}e^{-i \tau H_{A}} \ket{\psi}$.

\subsection{Formulation}
For the simulation of the Bosonic Hamiltonian in Eq~\eqref{eq:boson_hamiltonian}, we can follow a similar procedure of two-coloring (dividing the Hamiltonian into $H_A$ and $H_B$) and then applying rotation $R_{X}(2\sqrt{x}\tau)( = e^{-i\sqrt{x}\tau X})$ to the ancilla qubit based on the values of $\ket{x}$ for each $H_A$ and $H_B$. However, please note that unlike the example of the banded matrices, in our current case we need to incorporate the `square root' in  $R_{X}(2\sqrt{x}\tau) := e^{-i\sqrt{x}\tau X}$ in the exponential. 
In other words, our goal is to use QSP in two stages to transform $R_{X}(2x\tau)$ into $R_{X}(2\sqrt{x}\tau)$. Writing explicitly, the final rotation operator on the ancilla qubit will look like 
\be 
R_{X}(2\sqrt{x}\tau) = \mqty[\cos{(\sqrt{x}\tau)} & i\sin{(\sqrt{x}\tau)} \\ i\sin{(\sqrt{x}\tau)} & \cos{(\sqrt{x}\tau)}] .
\label{eq:rx_final}
\ee 
Note that, $ \cos{(\sqrt{x}\tau)} = \sum_{j=1}^{\infty}\frac{(x\tau^2)^{j}}{(2j)!}$. If we can block-encode $x\tau^2$, then we can apply GQSP to achieve our final goal of implementing Eq.~\eqref{eq:rx_final}.
So, we start from a $\frac{\pi}{2}$ shifted $R_X$ rotation 
\bea 
R_{X}\left(2(x\tau^2 - \frac{\pi}{2})\right) &=& \mqty[\cos{(x\tau^2 - \frac{\pi}{2})} & -i\sin{(x\tau^2 - \frac{\pi}{2})} \\ -i\sin{(x\tau^2 - \frac{\pi}{2})} & \cos{(x\tau^2 - \frac{\pi}{2})}] \notag \\
&=& \mqty[\sin{(x\tau^2 )} & i\cos{(x\tau^2 )} \\ i\cos{(x\tau^2)} & \sin{(x\tau^2 )}]
\label{eq:qsp_operator}
\eea
and use it with $U_1$ in eq.~\eqref{eq:U1}. The reason for a $\frac{\pi}{2}$-shifted rotation is that we could implement $\arcsin$ on $\sin{(x\tau^2)}$ using QSP and achieve the block encoding of $x\tau^2$ (see Eq.~\eqref{eq:polyarcsin_bound}). In other words, we define a new $\Tilde{U}_1$ such that,
\bea 
\Tilde{U}_1  \ket{x}\ket{f(x)} \left(  \alpha \ket{0}_{a} + \beta \ket{1}_{a} \right) &\rightarrow &
  \ket{x}\ket{f(x)}R_{X}(2x\tau^2 - \pi) \left(  \alpha \ket{0}_{a} + \beta \ket{1}_{a} \right)
\label{eq:U1a}
\eea 
and then apply iterated QSP to achieve the following,
\be 
\texttt{iQSP} \left[ R_{X}\left(2(x\tau^2 - \frac{\pi}{2})\right) \right] = \mqty[\cos{(\sqrt{x}\tau)} & i\sin{(\sqrt{x}\tau)} \\ i\sin{(\sqrt{x}\tau)} & \cos{(\sqrt{x}\tau)}] .
\label{eq:iqsp}
\ee 
By setting $a = \sin{(x\tau^2)}$ and following the GQSP convention in Eq.~\eqref{eq:gqsp} and from Eq.~\eqref{eq:qsp_operator} and ~\eqref{eq:iqsp} we can write, 
\bea 
P(a) &=& \cos{( \sqrt{ \arcsin(a)} )} \label{eq: P}\\
Q(a) &=& \sin{ (\sqrt{ \arcsin(a)})}\label{eq:Q}
\eea 
Since $P(a)$ is a composite function, we can use Lemma 70 in ref \cite{Gily_n_2019} to implement $\frac{2}{\pi}\arcsin(a)\equiv \frac{2}{\pi} x\tau^2$ first.  Subsequently, we will follow the convention where the superscripts will denote the number of qubits in the registers and the subscripts will denote the corresponding qubit registers. We use the polynomial expansion of $\frac{2}{\pi}\arcsin(a)$ in \cite{Gily_n_2019} up to $k$ terms  and construct a unitary $U_{\Phi}$ such that, 
\begin{equation*}
U_{\Phi}= \left( I_x \otimes e^{i \phi_{0} Z}   \right) \prod_{j=1}^{k}  \Tilde{U}_1 \left( I_x \otimes  e^{i \phi_{j} Z} \right).
\end{equation*}
  where  $\Phi$ is a vector of phases $\Phi \in \mathbb{R}^k $.  Since  $\arcsin(a)$ is not non-analytic at $a = \pm 1$, we limit our value of $a$ to $a \in [-1 + \delta, +1, 1 - \delta], \delta > 0$. Truncating $\arcsin(a)$ up to $k$ terms introduces an error $\epsilon_1$. $U_{\Phi}$ is then a $ \left( \frac{\pi}{2}, n+1 ,\epsilon_1 \right) $ block encoding of $\arcsin (a)$ such that,
 \begin{equation}
\abs\bigg{ \arcsin(a) - \frac{ \pi}{2} \left( \langle x | \otimes \bra{0}_{a} \right) U_{\Phi} \left( | x \rangle    \otimes \ket{0}_{a} \right) }_{a\in [-1+\delta,1-\delta]} \leq \epsilon_1
\label{eq:error_1}
\end{equation}
 This is shown in  Fig.~\ref{fig:be_ntau2}. From \cite{Gily_n_2019} we have,
\be
k = {\mathcal{O}\left(\frac{1}{\delta} \log(\frac{1}{\epsilon_1})\right)}
\label{eq:error_1a}
\ee 
for $\delta,\epsilon_1 \in ( 0,\frac{1}{2} ]$ . 
\begin{figure}[t!]
\begin{center}
\tiny
  \begin{quantikz}
 \lstick{$\ket{0}_a$} & \qw &\qw   &\qw  &\gate[3, nwires={3}, alternate]{\texttt{CMP}} &\ctrl{1} &\gate[1]{R_{X}(\phi_0)} &\gate[2, nwires={2}, alternate]{\Tilde{U}_1}  &\gate[1]{R_{X}(\phi_1)} &\qw &\qw & \ \ldots\ &\gate[2, nwires={2}, alternate]{\Tilde{U}_1} &\gate[1]{R_{X}(\phi_k)}   &\ctrl{1} &\gate[3, nwires={3}, alternate]{\texttt{CMP}} & \qw \\
\lstick{$\ket{x}$} &\qwbundle{n} &\qw  & \gate[2]{O_f} \qw  &\qw & \gate[2]{\texttt{SWP}}&\qw  &\qw &\qw &\qw &\qw   & \ \ldots\   &\qw &\qw & \gate[2]{\texttt{SWP}}\qw &\qw & \gate[2]{O_f}& \qw \\
\lstick{$\ket{y}$} &\qwbundle{n} &  \qw  &\qw   &\qw &\qw &\qw &\qw &\qw &\qw  &\qw  & \ \ldots\  &\qw &\qw &\qw &\qw &\qw &\qw \\
\end{quantikz}   
 \end{center}
\caption{Circuit for $\arcsin (a)$ on the ancilla using QSP. $U_1$ (defined in Eq.~\eqref{eq:U1}) contains the controlled $R_{X}$ gate. The  gate set $\{R_{X}(\phi_j), j = 0,..,k; U_1\}$ constitutes $U_{\Phi}$. }\label{fig:be_ntau2}
\end{figure}
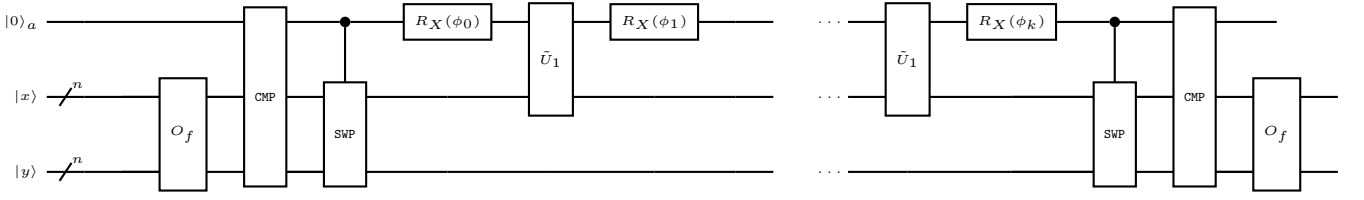
The value of $\delta$ sets an upper bound for the input qubit register $x$. Since $a \le 1-\delta$, we will have
 \be
 x_{max} \le \frac{\arcsin(1-\delta)}{\tau^2}
 \ee 

Our final step is to implement $P(a) = \sum_{j=1}^{\infty}\frac{(\arcsin(a))^{j}}{(2j)!}$. Traditional QSP cannot be directly used to implement $P(a)$. This is a drawback that QSP continues to face, and is a notable limitation in our construction of more general polynomial families. Notably, to construct any polynomial like $P(a)$ using the original QSP, one must divide it into two segments by separating its  even and odd components. These segments are then combined using linear combinations of unitaries (LCU), followed by a modified form of amplitude amplification on the resultant circuit. This process necessitates tripling the required operations for Hamiltonian simulation compared to implementing Hamiltonian dynamics directly. Additionally, while the original QSP assures the existence of parameter sets yielding polynomials meeting the specified constraints, determining these parameters has proven to be a complex  task.  This paper employs an advanced Generalized Quantum Signal Processing (GQSP) algorithm, which addresses the aforementioned limitations and establishes a broader framework for the primary objective of implementing Hamiltonian functions. 
Therefore, finally the transformation in Eq.~\eqref{eq: P} up to $l$-order truncated series can be achieved using GQSP by interleaving $U_{\Phi}$ and using in eq.~\eqref{eq:gqsp},
\begin{equation}
    U_{\vec{\theta}, \vec{\omega}} = \left( \prod_{j=1}^{l} R(\theta_j, \omega_j,0) 
    (\dyad{0}\otimes U_{\Phi} + \dyad{1}\otimes I) \right) R(\theta_0, \omega_0,\lambda)
\end{equation}
and using the circuit in Fig.~\ref{fig:full_circuit}.
\begin{figure}[t!]
\begin{center}
\tiny
 \begin{quantikz}
 \lstick{$\ket{0}$} & \qw &\qw &\qw  &\qw &\gate[1]{R(\theta_0,\omega_0,\lambda)} & \octrl{1} &\gate[1]{R(\theta_1,\omega_1,0)} & \octrl{1} &\qw &  \ldots & \octrl{1} &\gate[1]{R(\theta_l,\omega_l,0)} & \qw\\
 \lstick{$\ket{0}_a$} & \qw &\qw   &\qw  &\gate[3, nwires={3}, alternate]{\texttt{CMP}} &\ctrl{1} &\gate[2, nwires={2}, alternate]{U_{\Phi}}  &\qw &\gate[2, nwires={2}, alternate]{U_{\Phi}}   &\qw & \ \ldots\ &\gate[2, nwires={2}, alternate]{U_{\Phi}}    &\ctrl{1} &\gate[3, nwires={3}, alternate]{\texttt{CMP}} & \qw \\
\lstick{$\ket{x}$} &\qwbundle{n} &\qw  & \gate[2]{O_f} \qw  &\qw & \gate[2]{\texttt{SWP}}&\qw   &\qw &\qw &\qw   & \ \ldots\    &\qw & \gate[2]{\texttt{SWP}}\qw &\qw & \gate[2]{O_f}& \qw \\
\lstick{$\ket{y}$} &\qwbundle{n} &  \qw  &\qw   &\qw &\qw &\qw &\qw &\qw  &\qw  & \ \ldots\   &\qw &\qw &\qw &\qw &\qw &\qw\\
\end{quantikz}   
\end{center}
\caption{Quantum circuit for Implementing $P(a)$ up to $l$-order truncation. $U_{\Phi}$ is the block encoding for $\arcsin(a)(=x\tau^2$). The  gate set $\{R(\theta_j,\omega_j,\lambda), j = 0,..,l; U_{\Phi}\}$ constitutes $U_{\vec{\theta}, \vec{\omega}}$.}\label{fig:full_circuit}
\end{figure}
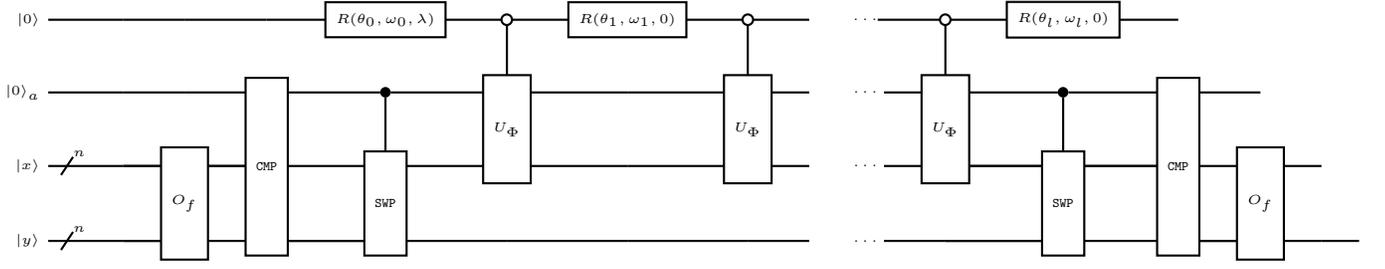

Denoting the $l$-order truncation error of $P(a)$ as $\epsilon_2$,  the phase angles generating a $P$ exist and can be efficiently calculated on a classical computer in time $ \mathcal{O} (poly(l, \log(1/ \epsilon_2)))$. So we can write,

\begin{equation}
\left| \cos{( \sqrt{ \arcsin(a)} )} -  P(a) \right| \leq \epsilon_2
\label{eq:error_2}
\end{equation}

The error due to truncation at the $l$ order is given by,
\begin{equation}
l = \mathcal{O}(\log(1/\epsilon_2))
\label{eq:error_2a}    
\end{equation} 
We repeat the above processes for $H_A$ and $H_B$ and combine them using Product Formula to simulate the full Hamiltonian.


\subsection{Error bound}
Let us call the polynomial transformation for $\arcsin(a)$ using a polynomial truncated up to $k$ terms for as $P_{1}^{(k)}(a)$, whereas the infinite series as $P_{1}(a)$. 
Total approximation error of $\arcsin(a)$ using a polynomial truncated up to $k$ terms and then $\cos{( \sqrt{ \arcsin(a)} )}$ truncated up to $l$ terms is given by,
\bea 
\epsilon &=& \sum_{j=0}^{l} \frac{(P_{1}(a))^j}{2j!} - \frac{(P_{1}^{(k)}(a))^j}{2j!} + \sum_{j=l}^{\infty}\frac{(P(a))^j}{2j!} 
\label{eq:total_error}
\eea 
Using a similar argument as presented in Eq 4.69 in Nielsen and Chuang ~\cite{nielsen2001quantum} we can write
\be 
\abs{U_m U_{m-1}... U_1 - V_{m} V_{m-1}... V_{1}} \leq \sum_{l=1}^{m}\abs{U_l - V_l}.
\ee 
where $U_l, V_l \in \mathbb{R}$ . 
In our case, choosing $U_m = U_{m-1}... =U_1 = P_{1}(a)$ and $V_{m} = V_{m-1}... =V_{1} = P_{1}^{(k)}(a)$ we have,
\bea 
\abs{(P_{1}(a))^j - (P_{1}^{(k)}(a))^j} &\leq& \sum_{l=1}^{j}\abs{P_{1}(a) -P_{1}^{(k)}(a) } \notag \\
&\leq& j\abs{P_{1}(a) -P_{1}^{(k)}(a) }
\label{eq:inequality}
\eea 
Now using the result of Eq.~\eqref{eq:inequality}, we replace in Eq.~\eqref{eq:total_error},
\bea 
&\le& \sum_{j=0}^{l} \frac{j\abs{P_{1}(a) - P_{1}^{(k)}(a)}}{2j!} + \sum_{j=l}^{\infty}\frac{(P(a))^j}{2j!} \notag \\
&\le&  \abs{P_{1}(a) - P_{1}^{(k)}(a)} + \sum_{j=l}^{\infty}\frac{(P(a))^j}{2j!} \equiv \epsilon_1 + \epsilon_2  
\eea 
In the last line we have used the results from Eq.~\eqref{eq:error_1} and Eq.~\eqref{eq:error_2}. 
Using the results from Eq.~\eqref{eq:error_1a} and Eq.~\eqref{eq:error_2a},
Assuming $\epsilon_1 = \epsilon_2 = \frac{\epsilon}{2}$, the total error bound is given in terms of $l$ and $k$,
\be 
k\times l = \mathcal{O} \left( \log^2(1/\epsilon) \right)
\ee 
 The error for using $p$-th order product formula for decomposing $e^{-iH\tau}$ into $e^{-iH_{A}\tau} e^{-iH_{B}\tau}$  scales as $\mathcal{O}(\tau^{p+1})$. 
 Hence the total error after running $N$ Trotter steps using $p$-th order product formula will be,
\bea
\epsilon_{tot} &=& N(\epsilon + \epsilon_{Trotter}) \notag \\
&=& N(\epsilon + \mathcal{O}(2^{p+1} \text{max}\{\norm{H_A},\norm{H_B}\}^{p+1}\tau^{p+1}) ) 
\eea 
where we have used from~\cite{childs2021theory} the following trotter bound which we will further take to be sufficiently small to ensure that the total error after $N$ steps is $\epsilon_{tot}$

\begin{equation}
    \epsilon_{Trotter} = \mathcal{O}(2^{p+1} \text{max}\{\norm{H_A},\norm{H_B}\}^{p+1}\tau^{p+1}) = \frac{\epsilon_{tot}}{2N}
\end{equation}
For a total simulation time $T = N\tau$, we can solve these equations to find an appropriate value of $N$ doing so yields
\bea 
N = \mathcal{O}\left(\frac{(\max\{\|H_A\|,\|H_B\|\} T)^{1+1/p}}{\epsilon^{1/p}} \right)
\eea
Then using the arguments made in~\cite{berry2007efficient} we can adaptively choose $p$ as a function of $\epsilon, T$ to achieve
\begin{equation}
    N= \mathcal{O} \left(\frac{(\max\{\|H_A\|,\|H_B\|\} T)^{1+o(1)}}{\epsilon^{o(1)}} \right)
\end{equation}
where $x^{o(1)}$ is a function that has sub-polynomial scaling with $x$ but not necessarily poly-logarithmic.  As $k,l$ are polylogarithmic functions of $1/\epsilon$,  the overall complexity is then
\begin{equation}
    klN = \widetilde{\mathcal{O}}\left(\frac{(\max\{\|H_A\|,\|H_B\|\} T)^{1+o(1)}}{\epsilon^{o(1)}} \right)
\end{equation}
Here $\widetilde{\mathcal{O}}(\cdot)$ denotes an asymptotic upper bound that neglects sub-dominant poly-logarithmic scalings.

The key point behind this calculation is that, while it is not as time-efficient as the best qubitization based approaches, it requires at most two ancillary qubits.  This means that this approach can be used to perform bosonic simulations on near-term hardware, which constitutes a substantial improvement over Toffoli-based approaches which can be prohibitively memory intensive for existing quantum devices.

\section{Conclusion}
In this paper, we provided a systematic paradigm for combining QSP sequences in order to construct polynomial transformations that would be difficult to achieve directly for traditional single qubit (SU(2)) Quantum Signal Processing. We first showed that we can recursively construct a form of multi-variate QSP within our framework by chaining together QSP sequences that allows the summing of rotation angles. We then provided a QSP construction that allowed us to multiply different signal unitary rotation angles by employing QSP to square the sum of different signal unitary rotation angles and subtracting off all other terms except this product. This gives us an elementary constructive approach to implement any multivariate polynomial function of a set of single qubit unitaries. 

The second application we considered involved applying our method of iterated QSP to the simulation of Bosonic quantum simulations. By taking advantage of our banded Bosonic Hamiltonian's sparsity, we showed how iterated QSP can be used to implement the complex arithmetic of Bosonic quantum simulation using only two ancillary qubits with relative ease. This approach opens up the possibility of near-term quantum simulation of Bosonic systems with cutoffs in occupation greater than 1.

At a high level, the biggest contribution of this work is to provide a concrete method for implementing multivariate quantum signal processing by reducing it to a special case of single variable QSP.  Further, the scaling for the single variable multiplication gadget is interestingly comparable to the complexity of ordinary multiplication.  However, the complexity of the algorithms that we provide are polynomial for low-order polynomial transformations, but as the degree of the polynomial increases the methods that we propose do not scale favorably with the degree of the polynomial considered and further work is needed to understand whether there are fundamental limitations to the scaling that can be achieved with quantum signal processing in the multivariate case or whether the constructions considered here are sub-optimal.  It is our hope that iterated block encodings such as the ones that we consider here form the beginning of such a discussion that will ultimately lead to highly practical methods for performing quantum arithmetic and simulations on quantum computers.
\section*{Acknowledgements}
This work was supported by the “Embedding QC into Many-body Frameworks for Strongly Correlated Molecular and Materials Systems” project, which is funded by the U.S. Department of Energy, Office of Science, Office of Basic Energy Sciences (BES), the Division of Chemical Sciences, Geosciences, and Biosciences (NG).
We thank Jack Ceroni for discussions on modular QSP and feedback on this work.
\bibliographystyle{unsrt}
\bibliography{bib}

\begin{thebibliography}{10}

\bibitem{Low_2019}
Guang~Hao Low and Isaac~L. Chuang.
\newblock Hamiltonian simulation by qubitization.
\newblock {\em Quantum}, 3:163, July 2019.

\bibitem{low2016methodology}
Guang~Hao Low, Theodore~J Yoder, and Isaac~L Chuang.
\newblock Methodology of resonant equiangular composite quantum gates.
\newblock {\em Physical Review X}, 6(4):041067, 2016.

\bibitem{gilyen2019quantum}
Andr{\'a}s Gily{\'e}n, Yuan Su, Guang~Hao Low, and Nathan Wiebe.
\newblock Quantum singular value transformation and beyond: exponential
  improvements for quantum matrix arithmetics.
\newblock In {\em Proceedings of the 51st Annual ACM SIGACT Symposium on Theory
  of Computing}, pages 193--204, 2019.

\bibitem{motlagh2023generalized}
Danial Motlagh and Nathan Wiebe.
\newblock Generalized quantum signal processing.
\newblock {\em arXiv preprint arXiv:2308.01501}, 2023.

\bibitem{rossi2022multivariable}
Zane~M Rossi and Isaac~L Chuang.
\newblock Multivariable quantum signal processing (m-qsp): prophecies of the
  two-headed oracle.
\newblock {\em Quantum}, 6:811, 2022.

\bibitem{Martyn_2021}
John~M. Martyn, Zane~M. Rossi, Andrew~K. Tan, and Isaac~L. Chuang.
\newblock Grand unification of quantum algorithms.
\newblock {\em PRX Quantum}, 2(4), December 2021.

\bibitem{alexis2024quantum}
Michel Alexis, Gevorg Mnatsakanyan, and Christoph Thiele.
\newblock Quantum signal processing and nonlinear fourier analysis.
\newblock {\em Revista Matem{\'a}tica Complutense}, pages 1--40, 2024.

\bibitem{low2019hamiltonian}
Guang~Hao Low and Isaac~L Chuang.
\newblock Hamiltonian simulation by qubitization.
\newblock {\em Quantum}, 3:163, 2019.

\bibitem{berry2024doubling}
Dominic~W Berry, Danial Motlagh, Giacomo Pantaleoni, and Nathan Wiebe.
\newblock Doubling efficiency of hamiltonian simulation via generalized quantum
  signal processing.
\newblock {\em arXiv preprint arXiv:2401.10321}, 2024.

\bibitem{childs2012hamiltonian}
Andrew~M Childs and Nathan Wiebe.
\newblock Hamiltonian simulation using linear combinations of unitary
  operations.
\newblock {\em arXiv preprint arXiv:1202.5822}, 2012.

\bibitem{childs2017quantum}
Andrew~M Childs, Robin Kothari, and Rolando~D Somma.
\newblock Quantum algorithm for systems of linear equations with exponentially
  improved dependence on precision.
\newblock {\em SIAM Journal on Computing}, 46(6):1920--1950, 2017.

\bibitem{Gily_n_2019}
András Gilyén, Yuan Su, Guang~Hao Low, and Nathan Wiebe.
\newblock Quantum singular value transformation and beyond: exponential
  improvements for quantum matrix arithmetics.
\newblock In {\em Proceedings of the 51st Annual ACM SIGACT Symposium on Theory
  of Computing}, STOC ’19. ACM, June 2019.

\bibitem{dong2022ground}
Yulong Dong, Lin Lin, and Yu~Tong.
\newblock Ground-state preparation and energy estimation on early
  fault-tolerant quantum computers via quantum eigenvalue transformation of
  unitary matrices.
\newblock {\em PRX Quantum}, 3(4):040305, 2022.

\bibitem{low2024quantum}
Guang~Hao Low and Yuan Su.
\newblock Quantum eigenvalue processing.
\newblock {\em arXiv preprint arXiv:2401.06240}, 2024.

\bibitem{dong2022infinite}
Yulong Dong, Lin Lin, Hongkang Ni, and Jiasu Wang.
\newblock Infinite quantum signal processing.
\newblock {\em arXiv preprint arXiv:2209.10162}, 2022.

\bibitem{rossi2023modular}
Zane~M. Rossi, Jack~L. Ceroni, and Isaac~L. Chuang.
\newblock Modular quantum signal processing in many variables, 2023.

\bibitem{Rossi_2022}
Zane~M. Rossi and Isaac~L. Chuang.
\newblock Multivariable quantum signal processing (m-qsp): prophecies of the
  two-headed oracle.
\newblock {\em Quantum}, 6:811, September 2022.

\bibitem{peng2023quantum}
Bo~Peng, Yuan Su, Daniel Claudino, Karol Kowalski, Guang~Hao Low, and Martin
  Roetteler.
\newblock Quantum simulation of boson-related hamiltonians: Techniques,
  effective hamiltonian construction, and error analysis.
\newblock {\em arXiv preprint arXiv:2307.06580}, 2023.

\bibitem{haener2018quantum}
Thomas Haener, Mathias Soeken, Martin Roetteler, and Krysta~M Svore.
\newblock Quantum circuits for floating-point arithmetic.
\newblock In {\em International Conference on Reversible Computation}, pages
  162--174. Springer, 2018.

\bibitem{németh2023variants}
Balázs Németh, Blanka Kövér, Boglárka Kulcsár, Roland~Botond Miklósi,
  and András Gilyén.
\newblock On variants of multivariate quantum signal processing and their
  characterizations, 2023.

\bibitem{mori2023comment}
Hitomi Mori, Keisuke Fujii, and Kaoru Mizuta.
\newblock Comment on "multivariable quantum signal processing (m-qsp):
  prophecies of the two-headed oracle", 2023.

\bibitem{Low_2017}
Guang~Hao Low and Isaac~L. Chuang.
\newblock Optimal hamiltonian simulation by quantum signal processing.
\newblock {\em Physical Review Letters}, 118(1), January 2017.

\bibitem{nielsen2001quantum}
Michael~A Nielsen and Isaac~L Chuang.
\newblock {\em Quantum computation and quantum information}, volume~2.
\newblock Cambridge university press Cambridge, 2001.

\bibitem{berry2007efficient}
Dominic~W Berry, Graeme Ahokas, Richard Cleve, and Barry~C Sanders.
\newblock Efficient quantum algorithms for simulating sparse hamiltonians.
\newblock {\em Communications in Mathematical Physics}, 270:359--371, 2007.

\bibitem{childs2021theory}
Andrew~M Childs, Yuan Su, Minh~C Tran, Nathan Wiebe, and Shuchen Zhu.
\newblock Theory of trotter error with commutator scaling.
\newblock {\em Physical Review X}, 11(1):011020, 2021.

\end{thebibliography}

\end{document}